\documentclass[11pt]{article}
\usepackage{amsmath,amssymb,amsfonts,latexsym,graphicx,amsthm}
\usepackage{fullpage,color}
\usepackage{url,hyperref}
\usepackage{comment}
\usepackage[linesnumbered,boxed,ruled,vlined]{algorithm2e}
\usepackage{framed}
\usepackage[shortlabels]{enumitem}
\usepackage{cleveref}
\usepackage{todonotes}
\usepackage[normalem]{ulem}
\usepackage{subcaption}
\usepackage[margin=1in]{geometry}
\usepackage{ifthen}
\usepackage{thm-restate}
\usepackage{xcolor}
\usepackage{stmaryrd}

\newtheorem{theorem}{Theorem}[section]
\newtheorem{lemma}{Lemma}[section]

\newtheorem{corollary}{Corollary}[section]
\newtheorem{claim}{Claim}[lemma]

\newtheorem{observation}{Observation}[section]

        {\medskip}

\newcommand{\eps}{\epsilon}
\newcommand{\hide}[1]{}

\newcommand{\opt}{\mathsf{opt}}
\newcommand{\OPT}{\mathsf{OPT}}

\newcommand{\mst}{\mathsf{MST}}

\newcommand{\wts}{\mathbf{w}}

\newcommand{\iter}{\mathsf{iter}}
\definecolor{BrickRed}{rgb}{0.8, 0.25, 0.33}
\def\EMPH#1{\emph{\textcolor{BrickRed}{#1}}}

\newcommand{\proj}{\mathrm{proj}}
\newcommand{\slope}{\mathrm{slope}}
\newcommand{\slack}{\mathrm{slack}}

\begin{document}

	\title{Approximating Euclidean Shallow-Light Trees}
	\author{
		Hung Le\thanks{University of Massachusetts Amherst, \texttt{hungle@cs.umass.edu}}\and
		Shay Solomon\thanks{Tel Aviv University, \texttt{solo.shay@gmail.com}}\and	
		Cuong Than\thanks{University of Massachusetts Amherst, \texttt{cthan@umass.edu}}\and
		Csaba D. T\'oth\thanks{California State University Northridge and Tufts University, \texttt{csaba.toth@csun.edu}}\and
		Tianyi Zhang\thanks{Nanjing University, \texttt{tianyiz25@nju.edu.cn}}
	}	
	\date{}
	\maketitle

\thispagestyle{empty}
For a weighted  graph 
$G = (V, E, w)$ and a designated source vertex $s \in V$, a spanning tree that simultaneously approximates a {\em shortest-path tree} w.r.t.\ source $s$ and a {\em minimum spanning tree} is called a {\em shallow-light tree (SLT)}. Specifically, an $(\alpha, \beta)$-SLT of $G$ w.r.t.\ $s \in V$ is a spanning tree of $G$ with {\em root-stretch} $\alpha$ (preserving all distances between $s$ and the other vertices up to a factor of $\alpha$)
and {\em lightness} $\beta$ (its {\em weight}  is at most $\beta$ times the weight of a minimum spanning tree of $G$).

It was shown in the early 90s that (1) for any graph and any 
$\epsilon > 0$, there is a $(1 + \eps, O(1/\epsilon))$-SLT w.r.t.\ any source, and (2) there {\em exist} graphs for which $\beta = \Omega(1/\eps)$
for any $(1+\eps,\beta)$-SLT.

The focus of this work is on SLTs in \EMPH{low-dimensional Euclidean spaces}, which are of special interest for some applications of SLTs, 
in geometric network optimization problems
such as VLSI circuit design.
The aforementioned existential lower bound applies to Euclidean plane, as well. It was shown more than a decade ago that (1) by using  {\em Steiner points}, one can reduce the lightness bound from $O(1/\eps)$ to $O(\sqrt{1/\eps})$, and (2) there exist point sets in Euclidean plane 
for which $\beta = \Omega(\sqrt{1/\eps})$ for
any Steiner $(1+\eps,\beta)$-SLT.

These \EMPH{tight existential bounds} for the Euclidean case 
yield \EMPH{approximation factors} of $O(1/\eps)$ and $O(\sqrt{1/\eps})$ on the minimum weight of any non-Steiner and Steiner tree with root-stretch $1+\eps$, respectively.
Despite the large body of work on SLTs, 
the basic question of whether a better approximation algorithm 
exists was left untouched to date, and this holds in any  graph family. 
This paper makes a first nontrivial step towards this question by presenting two 
\EMPH{bicriteria approximation}  algorithms.
For any $\eps>0$, a  set $P$ of $n$ points in constant-dimensional Euclidean space and a source $s\in P$, our first (respectively, second) algorithm returns, in
$O(n \log n \cdot {\rm polylog}(1/\eps))$ time,  
a non-Steiner (resp., Steiner) tree with root-stretch $1+O(\eps\log \eps^{-1})$ and weight at most $O(\opt_{\eps}\cdot \log^2 \eps^{-1})$ (resp., 
$O(\opt_{\eps}\cdot \log \eps^{-1})$), where  $\opt_{\eps}$ denotes the minimum weight of a non-Steiner (resp., Steiner) tree with root-stretch $1+\eps$.

\clearpage

\thispagestyle{empty}   

\thispagestyle{empty}
\setcounter{tocdepth}{3}
\tableofcontents
\clearpage

\clearpage
\setcounter{page}{1}
\section{Introduction}
\label{sec:intro}

A \emph{shortest-path tree (SPT)} of an undirected edge-weighted $n$-vertex graph $G = (V,E,\wts)$ with respect to a designated {\em source} or {\em root} vertex $s \in V$, denoted by $\mathsf{SPT}(G,s)$ 
is a spanning tree $T$ rooted at $s$ that 
preserves all distances from $s$, i.e., for every vertex $v \in V$,
the distance $d_T(s,v)$ between $s$ and $v$ in $T$ equals their distance $d_G(s,v)$ in $G$.  
For a parameter $\alpha \ge 1$, an {\em $\alpha$-shallow tree ($\alpha$-ST)} is a spanning tree $T$ of $G$ of \emph{root-stretch} at most $\alpha$, i.e., for every $v \in V$, $d_T(s,v) \le \alpha \cdot d_G(s,v)$. A \emph{minimum spanning tree (MST)} of $G$, denoted by $\mst(G)$, 
is a spanning tree $T$ of $G$ of minimum weight. 
For a parameter $\beta \ge 1$, a {\em $\beta$-light tree ($\beta$-LT)} is a spanning tree $T$ of $G$ of \emph{lightness} $\beta$, i.e., $\wts(T) \le \beta \cdot \wts(\mst(G))$.
The SPT and the MST, including their approximate versions, are among the most fundamental graph constructs and have been extensively studied over decades.

A single tree that simultaneously approximates the SPT and the MST is called a \emph{shallow-light tree (SLT)}. For a pair of parameters $\alpha, \beta \ge 1$, an $(\alpha,\beta)$-SLT of graph $G$ w.r.t.\ a designated source $s \in V$
is a spanning tree of $G$ that is both an $\alpha$-ST and a $\beta$-LT.
The notion of SLTs was introduced in the pioneering works of 
Awerbuch et al.~\cite{ABP90,ABP91} and Khuller et al.~\cite{KRY95} (see also \cite{CKRSW92}). They showed that for every $\eps > 0$, a $(1+\eps,O(\frac{1}{\eps}))$-SLT can be constructed in linear time for every graph $G$ if an SPT and an MST are given. Khuller et al.~\cite{KRY95}  
also showed that this tradeoff is tight, by presenting a \emph{planar} graph for which $\beta = \Omega(\frac{1}{\eps})$ for any $(1+\eps,\beta)$-SLT.

The balance between the useful properties of an MST, which provides a light-weight network, and of an SPT, which provides short paths from a designated source to all other vertices, 
has led to a wide variety of applications across diverse domains. 
This includes applications in routing \cite{AHHKK95,ChenY20,JDHLG01,KPP02,LiQCM021,SS97,WCT02} and in network and VLSI-circuit design \cite{CKRSW91,CKRSW92,HeldR13,SCRS97},
for data gathering and dissemination tasks in overlay networks
\cite{KhazraeiH20,KV01,VWFME03}, in the message-passing model of distributed computing \cite{ABP90,ABP91},
and in wireless and sensor networks \cite{RW04,BDS04,CBVW06,LLLD06,LL07,SS10}.
In addition to their direct applications, SLTs are used as building blocks in other related graph structures,
such as \emph{light approximate routing trees} \cite{WCT02},  \emph{shallow-low-light trees} \cite{DES09tech,ES11},
\emph{light spanners} \cite{ABP91,Peleg00}, and others \cite{SCRS97,LLLD06,LL07}. 
In particular, in real-world applications, such as VLSI-design and wireless communication networks, the vertices are embedded in Euclidean space, and the edge weights correspond to the metric distances between the nodes.

\smallskip\noindent\textbf{Low-dimensional Euclidean spaces.~}
Khuller et al.\ \cite{KRY95} asked whether a better construction of SLTs can be achieved in Euclidean plane, \EMPH{which is the focus of this work}.  
Euclidean space $\mathbb{R}^d$, $d \ge 1$, can be modeled as a complete edge-weighted graph $G = (V,E,\wts)$ induced by a finite set $P$ of points in $\mathbb{R}^d$ with $V = P$, $E = \binom{P}{2}$, and $\wts = \|\cdot\|_2$.
Elkin and Solomon~\cite{ElkinS15} showed that the upper bound of 
$(1+\eps,O(\frac{1}{\eps}))$-SLTs in general graphs \cite{ABP90,ABP91,KRY95} is asymptotically tight even in Euclidean plane:
For a set $C$ of $\lceil1/\eps\rceil$ evenly spaced points on a circle, any $(1+\eps,\beta)$-SLT for $C$ for any source $s\in C$ must have $\beta = \Omega(\frac{1}{\eps})$.
Solomon~\cite{Solomon15} showed that allowing \emph{Steiner points} lead to substantial improvement in Euclidean plane: For every set $P\subset \mathbb{R}^2$ and source $s\in P$, one can construct a \emph{Steiner} $(1+\eps,O(\sqrt{1/\eps}))$-SLT in linear time. Moreover, this bound is asymptotically tight: 
For the same set $C$ of $\lceil 1/\eps\rceil$ evenly spaced points on a circle (in fact, here $\lceil \sqrt{1/\eps}\rceil$ evenly spaced points suffice), any Steiner $(1+\eps,\beta)$-SLT for $C$ for any source $s\in C$ must have $\beta = \Omega(\sqrt{1/\eps})$ \cite{ElkinS15,Solomon15}. 

\smallskip\noindent\textbf{Approximation algorithms and hardness.} 
The aforementioned results provide \EMPH{tight existential bounds} on the tradeoff between root-stretch and lightness of SLTs in general graphs as well as in planar graphs and in Euclidean plane; moreover, as mentioned above, tight bounds were established also for \emph{Steiner} SLTs in Euclidean plane. 
However, these algorithms do not necessarily provide instance-optimal SLTs: In \Cref{sec:LB}, we present point sets $P\subset \mathbb{R}^2$ for which 
any previous SLT algorithm in~\cite{ABP90,ABP91,KRY95} returns a $(1+\eps)$-ST of weight $O(\frac{1}{\eps})\cdot \wts(\mst(P))$, Solomon~\cite{Solomon15} constructs a Steiner $(1+\eps)$-ST of weight $O(\sqrt{1/\eps})\cdot \wts(\mst(P))$, but the minimum weight of a $(1+\eps)$-ST is only $O(1)\cdot \wts(\mst(P))$.
Despite the large body of work on SLTs, very little is known about SLTs from the perspective of optimization and approximation algorithms. 

In the \EMPH{$(1+\eps)$-SLT problem}, we are given a parameter $\eps \ge 0$ 
and an edge-weighted graph $G$, and the goal is to find an $(1+\eps)$-ST for $G$ of minimum weight. 
Khuller et al.~\cite{KRY95} showed that for any $\eps > 0$, the $(1+\eps)$-SLT problem is NP-hard (via a reduction from 3SAT), while the case $\eps = 0$ can be solved in near-linear time. Cheong and Lee~\cite{CheongL13} showed that it is NP-hard in Euclidean plane, as well (via a reduction from Knapsack). A $\kappa$-approximation algorithm for the problem should return a $(1+\eps)$-ST for $G$ whose weight is at most $\kappa$ times that of a minimum weight $(1+\eps)$-ST. One can also consider a \EMPH{bicriteria approximation}: a $(\kappa_1,\kappa_2)$-approximation for the problem should return a $(1+\kappa_1 \cdot \eps)$-ST for $G$ whose weight is at most $\kappa_2$ times that of a  minimum-weight $(1+\eps)$-ST.

The \EMPH{tight existential bounds}, mentioned above, yield \EMPH{approximation factors} of $O(1/\eps)$ and $O(\sqrt{1/\eps})$, respectively, for the $(1+\eps)$-SLT problem
on general edge-weighted graphs and in Euclidean plane, respectively. 
To the best of our knowledge, no other approximation algorithm or hardness result is known for this problem, even for basic graph families such as the complete graph with Euclidean edge weights. (In \Cref{sec:related} we discuss a related problem, for which both approximation algorithms and hardness results are known.)

In Euclidean spaces, one can define the \EMPH{Steiner $(1+\eps)$-SLT problem}: For a parameter $\eps \ge 0$ and an input point set $P \subset \mathbb{R}^d$, $d\in \mathbb{N}$, 
the goal is to find a Steiner $(1+\eps)$-ST for $P$ of minimum weight. 
We note that a minimum weight Steiner $(1+\eps)$-ST may be significantly lighter than a minimum weight non-Steiner $(1+\eps)$-ST. For example, for a set $C$ of $\lceil 1/\eps\rceil$ evenly spaced points on a circle, the ratio between the weights of minimum weight non-Steiner and Steiner $(1+\eps)$-STs is $\Theta(\sqrt{1/\eps})$. 

\subsection{Our Contribution}
We provide a bicriteria approximation for the $(1+\eps)$-SLT problem, where $0 < \eps \ll 1$ is an arbitrary parameter, in any constant-dimensional Euclidean space. 
(We shall assume throughout that $\eps$ is a sub-constant parameter. If $\eps>0$ is a constant, the existential upper bound of $(1+\eps,O(\frac{1}{\eps}))$-SLTs already provides a constant approximation in linear time.) 

\begin{theorem}\label{thm:Steiner}
There is an $O(n \log n \cdot \mathrm{polylog}(\eps^{-1}))$-time algorithm that, given $\eps>0$, a finite set $P$ of $n$ points in Euclidean plane, including a source $s\in P$, returns a Steiner $(1+O(\eps\log \eps^{-1}))$-ST of weight at most $O(\opt_{\eps}\cdot \log \eps^{-1})$, where $\opt_{\eps}$ denotes the minimum weight of a Steiner $(1+\eps)$-ST.
The result extends without any loss in parameters to Euclidean space $\mathbb{R}^d$, for any constant $d\geq 3$.
\end{theorem}

Interestingly, our bicriteria approximation algorithm of \Cref{thm:Steiner} 
incurs the same $O(\log \eps^{-1})$ ratio for both the stretch approximation (to the additive $\eps$ term) and the weight approximation.
With some additional effort and another $\log\eps^{-1}$ factor in the weight approximation ratio, our result generalizes to the setting without Steiner points in the plane. 

\begin{theorem}\label{thm:main}
There is an $O(n \log n \cdot {\rm polylog}(\eps^{-1}))$-time algorithm that, given $\eps>0$, a finite set $P$ of $n$ points in Euclidean plane, including a source $s\in P$, returns an $(1+O(\eps \log \eps^{-1}))$-ST of weight at most $O(\opt_{\eps}\cdot \log^2\eps^{-1})$, where $\opt_{\eps}$ denotes the minimum weight of an $(1+\eps)$-ST.
The result extends to Euclidean space $\mathbb{R}^d$, for any constant $d\geq 3$, with 
approximation ratio increasing by a factor of $O(\log \eps^{-1})$  
and the running time increasing by a factor of $\mathrm{poly}(\eps^{-1})$.
\end{theorem}

To complement our results, we show in \Cref{sec:LB} that the approximation ratio of our algorithms (with or without using Steiner points) is significantly better than the state-of-the-art algorithms at the instance level. Specifically, we design point sets in Euclidean plane for which any previous algorithm returns a $(1+\eps)$-ST of approximation ratio at least $\Omega(\sqrt{1/\eps})$ with Steiner points and $\Omega\left(\frac{1}{\eps}\right)$ without Steiner points. 

\begin{theorem}\label{thm:easy-LB}
    For every $\eps>0$, there exists a set $P\subseteq \mathbb{R}^2$ and a source $s\in P$ such that the minimum weight of a $(1+\eps)$-ST (resp., Steiner $(1+\eps$)-ST) is 
    $O(1)\cdot \wts(\mst(P)),$
    but any previous algorithm in~\cite{ABP90,ABP91,KRY95}
    returns an $(1+\eps)$-ST of weight $\Omega(\frac{1}{\eps})\cdot \wts(\mst(P))$,
    and the algorithm in~\cite{Solomon15}
    returns a Steiner $(1+\eps)$-ST of weight $\Omega(\sqrt{1/\eps})\cdot \wts(\mst(P))$.
\end{theorem}

To prove \Cref{thm:Steiner} and \Cref{thm:main}, we reduced the problem to a set of centered $\eps$-net in a region in a cone with aperture $\sqrt{\eps}$, within $\Theta(1)$ distance from the root. The classical lower-bound construction for this problem consists of a set $P$ of uniformly distributed points along a circle of unit radius centered at the root $s$. However, if $P$ is the subset of points in a cone of angle $\sqrt{\eps}$, then there exists a Steiner $(1+\eps)$-ST of weight $O(1)\cdot \wts(\mst(P))$~\cite{Solomon15}. This raises the question: What is the maximum lightness of a Steiner $(1+\eps)$-ST for points in a cone of aperture $\sqrt{\eps}$? We give a lower bound on the maximum lightness of a minimum-weight Steiner $(1+\eps)$-ST for points in a cone of aperture $\sqrt{\eps}$ in \Cref{sec:LB-sector}.

\subsection{Related Work} \label{sec:related}

Elkin and Solomon~\cite{ElkinS15} studied the power of Steiner points for SLTs in general metric spaces. The Steiner points in the construction of \cite{ElkinS15} are not part of the input metric, and the only restriction on the Steiner points is to {\em dominate} the metric distances: for any pair of non-Steiner points,
their tree distance should be at least as large as their metric distance. Using such ``out of nowhere'' Steiner points, \cite{ElkinS15} constructed Steiner $(1+\eps,O(\log \eps^{-1}))$-SLTs (and also an SPT with lightness $O(\log n)$), and they also showed that this tradeoff is tight, by presenting a metric space for which $\beta = \Omega(\log \eps^{-1})$ for any Steiner $(1+\eps,\beta)$-SLT, for any $\eps = \Omega(1/n)$.

There is also a large body of work on light $(1+\eps)$-spanners in low-dimensional Euclidean and doubling metrics \cite{BorradaileLW19,bhore2022euclidean,BuchinRS25,le2022truly}, including approximation algorithms for the minimum weight $(1+\eps)$-spanner~\cite{althofer1993sparse,DKR15,DZ16,KP94,le2024towards,LSTTZ26}. 
On the one hand, an $(\alpha,\beta)$-SLT of a $(1+\eps)$-spanner is a $(\alpha(1+\eps),\beta(1+\eps))$-SLT of the original graph, and a sparse spanner may help
reduce computational overhead. On the other hand, a $(1+\eps)$-spanner provides $1+\eps$ stretch between all pairs of vertices: Every spanner algorithm critically exploits this property, by constructing spanners in increasing scales, where larger scales can inductively rely on smaller scales. In contrast, SLTs do not have this property, and so recent advances on spanner algorithms cannot be adapted to SLTs. 

Gudmundsson et al.~\cite{GudmundssonMU21} considered \EMPH{bounded-degree SLTs}.
For every finite metric space of doubling dimension $k$ and every integer $b\geq 2$, they constructed an $(O(1),\max\{O(k)/\log b, O(1)\})$-SLT of maximum degree $b$. 
However, the maximum degree of Euclidean $(1+\eps)$-STs and Steiner $(1+\eps)$-STs, is unbounded as $\eps$ decreases. For a set $C$ of $\lceil 1/\eps\rceil$ evenly spaced points on a circle, mentioned above, the maximum degree of every $(1+\eps)$-ST is $\Omega(\frac{1}{\eps})$, and the maximum degree of every Steiner $(1+\eps)$-ST is $\Omega(\sqrt{1/\eps})$. 

Kortsarz and Peleg~\cite{kortsarz1997approximating} introduced the \EMPH{$D$-MST} for a set a set $T\subset V$ of terminals in a graph $G=(V,E)$ as Steiner tree for $T$ of weight $O(\mst(G))$ and \emph{diameter} at most $D$. Later in the literature, the $D$-MST is often called \emph{shallow-light tree}, where the term \emph{shallow} refers to the diameter bound. The optimization problem of minimizing the weight of a $D$-MST was studied extensively~\cite{ChimaniS15,hajiaghayi2009approximating,NaorS97}. 
This alternate version of the shallow-light tree problem is clearly NP-hard since it generalizes the classical {\em minimum Steiner tree problem}. In fact, it was shown that for any $D>0$, there is no $(\ln n-\epsilon)$-approximation of the minimum weight $D$-MST in graphs with unit edge costs \cite{bar2001generalized} unless $\mathrm{NP}\subseteq \mathrm{DTime}(n^{\log\log n})$. On the algorithmic side, there is an exact algorithm with runtime $O(3^{|S|nD})$ where $S$ denotes the terminal set \cite{guo2012parameterized}. If we allow approximations, a polynomial time algorithm was shown in \cite{kortsarz1997approximating} with $\min\{D\log n, n^\epsilon\}$-approximation of the minimum weight tree. Allowing a bicriteria approximation, which relaxes the requirement on the radius, it was shown that one can compute in polynomial time an $O(\log^3n)$-approximation of the minimum weight $D$-MST by allowing an $O(D\log n)$ radius bound \cite{hajiaghayi2009approximating,KhaniS16}; 
when the edge costs and weights are linearly related, one can obtain $O(1)$-approximation with $O(D)$ diameter bound \cite{guo2014approximating}.

\subsection{Technical Overview}

Given a set $P$ of $n$ points in the plane, including a source $s\in P$, and a parameter $\eps>0$, we describe $O(n\log n\cdot \mathrm{polylog}(\eps^{-1}))$-time algorithms to construct a (Steiner) $(1+\eps\cdot \log \eps^{-1})$-ST rooted at $s$, and then analyze its weight compared to the minimum weight $(1+\eps)$-ST rooted at $s$. We note that, since $(1+\eps)$-STs do not have a recursive substructure, the stretch between two arbitrarily points in $P$ may be unbounded. Yet, we can apply a divide-and-conquer strategy by clustering nearby points together, based on their position w.r.t.\ the source $s$. 

In \Cref{ssec:nonSteinerReduction}, 
we partition the plane into trapezoid tiles, and show that the union of bicriteria approximate SLTs for the point sets the tiles is a bicriteria approximation for the entire point set for both the Steiner and non-Steiner settings (\Cref{thm:nonSteiner-reduction}). We construct a tiling based on geometric considerations. 
The diameter of each tile $\tau$ is proportional to the distance ${\rm dist}(s,\tau)$, and the shape of $\tau$ is roughly $A\times (A\cdot \sqrt{\eps})$ for $A={\rm dist}(s,\tau)$; see \Cref{fig:overview}. That is, we choose the aspect ratio of every tile to be roughly $\sqrt{\eps}$ for the following reason: The triangle inequality implies that every $ps$-path of weight at most $(1+\eps)\cdot d(p,s)$ lies in an ellipse $\mathcal{E}_{ps}$ with foci $p$ and $s$, and aspect ratio roughly $\sqrt{\eps}$; see~\Cref{sec:pre}. Therefore, the union of all ellipses $\mathcal{E}_{ps}$, for all points $p \in P\cap \tau$, will be similar to the tile $\tau$ in the sense that the aspect ratio of its bounding box is roughly $\sqrt{\eps}$. The shape of the tiles is crucial for the proof of the reduction (\Cref{thm:nonSteiner-reduction}).

For all points in a tile $P\cap \tau$, the distance $d(p,s)$ to $s$ is the same up to constant factors. We can further partition the set of points in each tile into cluster by approximating $d(p,s)$ up to an $(1+\eps)$-factor. Recall that a classical \emph{$\eps$-net} in a metric space $(X,d)$ is a set $N\subset X$ such that the points in $N$ are at least $\eps$ distance apart, and the $\eps$-neighborhood of every point $x\in X$ contains a net point in $N$. In \Cref{ssec:net}, we define a \EMPH{centered $\eps$-net} $N$, where points $a,b\in N$ are at least $\eps \cdot \max\{d(a,s),d(b,s)\}$ apart, and the $(\eps\cdot d(p,s))$-neighborhood of every point $p\in P$ contains a net point in $N$.  We show that a bicriteria approximate STs for a centered $\eps$-net can be extended to bicriteria ST for the entire point set, using $O(1)$-spanners in the neighborhoods of the net points (\Cref{lem:net}). Interestingly, we reduce the $(1+\eps)$-SLT problem for $P$ to a variant of the Steiner $(1+\eps)$-SLT problem for the net $N\subset P$, where all Steiner points must be in the original set $P$. We note that although the reduction steps in \Cref{sec:reduction} are new and essential to our approach, they are based mainly on standard techniques.

\begin{figure}[htbp]
    \centering
    \includegraphics[width=.9\textwidth]{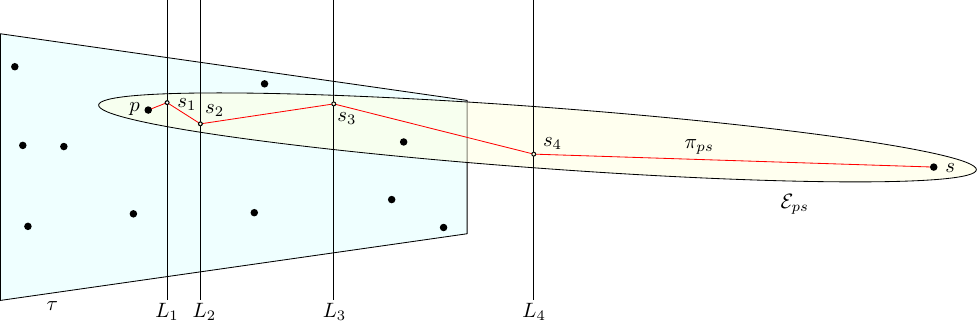}
    \caption{Points in a trapezoid tile $\tau$, and Steiner points $s_1,\ldots ,s_4\in \mathcal{E}_{ps}$ on parallel lines.}
    \label{fig:overview}
\end{figure}

The core technical contributions of our work appear in \Cref{sec:sector,sec:general-sector}, where we construct a Steiner ST for a centered $\eps$-net in \Cref{sec:sector} and then extend the construction to the non-Steiner setting in \Cref{sec:general-sector}. The Steiner setting is easier to work with because we can control the location of Steiner points.
We begin with a brief overview of the Steiner construction; refer to \Cref{fig:overview}. From the perspective of a single point $p\in P\cap \tau$, the construction is similar to the Steiner SLT construction by \cite{Solomon15}, which gave an existentially tight bound: We choose Steiner points $s_1, s_2,\ldots ,s_{k-1}$ in the ellipse $\mathcal{E}_{ps}$ on parallel lines $L_1,\ldots , L_{k-1}$ at distance $4^i \eps\cdot d(p,s)$ from $p$, where  $k=O(\log \eps^{-1})$. Solomon~\cite{Solomon15} shows that one can carefully choose Steiner points so that the stretch of the path $\pi_{ps}=(p,s_1,\ldots , s_{k-1},s)$ is at most $1+\eps$. However, geometric calculations show that even if we choose arbitrary points $s_i\in L_i\cap \mathcal{E}_{ps}$ for $i=1,\ldots , k-1$, then each edge $s_{i-1}s_i$ still contributes only $O(\eps)$ to the stretch (more precisely, $\wts(s_{i-1}s_i)$ exceeds the length of its orthogonal projection to the line $ps$ by $O(\eps)\cdot d(p,s)$). In other words, arbitrary Steiner points $s_i\in L_i\cap \mathcal{E}_{ps}$, for $i=1,\ldots , k-1$, guarantee a root-stretch $1+O(\eps\cdot \log \eps^{-1})$. 

For a point set $P_\tau = P\cap \tau$ in a tile $\tau$, we follow the above strategy, but we synchronize the lines $L_1,\ldots ,L_k$ chosen for different points in $P_\tau$. Then each line $L_i$ corresponds to many points $p\in P_\tau$, and intersects their ellipses $\mathcal{E}_{ps}$. We use a \EMPH{minimum hitting set} for the intervals $L_i\cap \mathcal{E}_{ps}$, to choose the minimum number of Steiner points that serve all associated ellipses. The weight analysis uses the fact 
that each ellipse $\mathcal{E}_{ps}$ contains a $ps$-path that crosses all lines $L_1,\ldots , L_k$.

In \Cref{sec:general-sector}, we adapt the Steiner ST algorithm to the non-Steiner setting. However, both the algorithm design and its analysis are more challenging. Instead of creating Steiner points $s_i$ in a line $L_i$ of our choice, now all points $s_i$ must be in $P$. 
We use the lines $L_1,\ldots , L_k$ to cover the ellipse $\mathcal{E}_{ps}$ with axis-aligned rectangles whose corners are on two consecutive lines $L_i$ and $L_{i+1}$; and then 
choose \emph{minimum hitting sets} from $P$ for the \emph{nonempty} rectangles. Some of the rectangles $R_i$ might be \emph{empty} (i.e., $R_i \cap P=\emptyset$). This means that we cannot choose a point $s_i\in P$ in $R_i$ for some $ps$-path (our algorithm simply skips $R_i$), but it also means that an optimal ST $\OPT$ does not have any vertices in $R_i$, to the $ps$-path in $\OPT$ contains an edge that traverses $R_i$ whose weight is proportional to the width of $R_i$. This is a crucial observation for the weight analysis. The root-stretch analysis also requires more work in the non-Steiner setting: In the Steiner case, the Steiner points $s_{i-1}$ and $s_i$ are on the lines $L_{i-1}$ and $L_i$, so we can control the distance between $s_{i-1}$ and $s_i$. However, when are limited to points $s_{i-1},s_i\in P$ in rectangles $R_{i-1}$ and $R_{i}$, it is possible that $s_{i-1}$ and $s_i$ are too close to each other, and their contribution to the root-stretch is too large. In such cases, we modify the $ps$-paths by skipping $s_{i-1}$ or $s_i$. This ensures that the distances between consecutive points of the $ps$-paths are sufficiently large, and we prove that the weight increases by at most a constant factor (\Cref{lem:skipping}).  In \Cref{sec:dspace}, we show that our algorithms (for both the Steiner and non-Steiner settings) extend naturally to higher dimensions using cone partitioning and approximate high-dimensional hitting sets. For the Steiner version, the approximation guarantee remains unchanged. However, the non-Steiner algorithm incurs an additional $\log(\varepsilon^{-1})$ factor in its approximation ratio. Its running time also increases by an additional $\mathrm{poly}(\varepsilon^{-1})$ factor.

\section{Preliminaries}
\label{sec:pre}

Let $p,s\in \mathbb{R}^2$, and $\eps>0$. If $\pi_{ps}$ is a $ps$-path of weight at most $(1+\eps)\cdot d(p,s)$, then every point $q\in \pi_{ps}$, we have $d(p,q)+d(q,s)\le\wts(\pi_{ps})\leq (1+\eps)\cdot d(p,s)$. This implies that $\pi_{ps}$ is contained in the ellipse $\mathcal{E}_{ps}$ with foci $p$ and $s$, and major axis $(1+\eps)d(p,s)$; see \Cref{fig:ellipse1}. The ellipse $\mathcal{E}_{ps}$ is contained in a rectangle spanned by its major and minor axes. The length of its minor axis is $\sqrt{(1+\eps)^2-1}\cdot d(p,s) = \sqrt{2\eps+\eps^2}\cdot d(p,s)< 2\sqrt{\eps}\cdot d(p,s)$ if $\eps<2$.

\begin{figure}[htbp]
    \centering
    \includegraphics[width=.65\columnwidth]{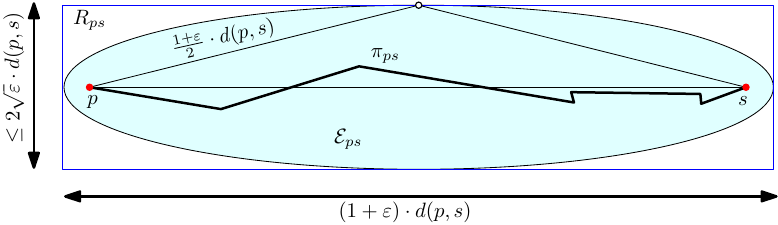}
    \caption{The ellipse $\mathcal{E}_{ps}$ with foci $p$ and $s$ and major axis of length $(1+\eps)\cdot d(p,s)$.}
    \label{fig:ellipse1}
\end{figure}

For analyzing the weight of an SLT, we consider a $ps$-path $\pi_{ps}$ as a polyline (i.e., a subset of the plane). In particular, for any region $R\subset \mathbb{R}^2$, the intersection $\pi\cap R$ is the a part of the polyline contained in $R$. 
For a line $L$, we denote by $L^-$ and $L^+$ the two open halfplanes bounded by $L$.
We make use of the following easy observation:

\begin{observation}\label{lem:charging}
    Let $L$ be a line that separates $p$ and $s$; see \Cref{fig:ellipse2} such that $p\in L^-$ and $s\in L^+$. For every $ps$-path $\pi_{ps}$, we have $\wts(\pi_{ps}\cap L^-)\geq d(p,L)$.
\end{observation}
\begin{proof}
Since $p\in L^-$ and $s\in L^+$, the path $\pi_{ps}$ must cross the line $L$. 
Let $q$ be the first point along $\pi_{ps}$ that lies in $L$. Then the subpath of $\pi_{ps}$ from $p$ to $q$ lies in $L^-$ and its length is at least $d(p,L)$, as claimed. 
\end{proof}

\begin{figure}[htbp]
    \centering
    \includegraphics[width=.65\columnwidth]{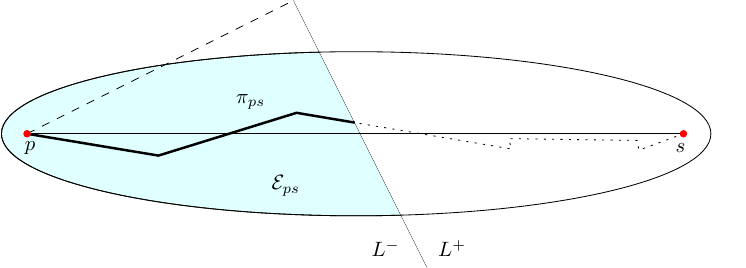}
    \caption{A part of the path $\pi_{ps}$ in the halfplane $L^-$.}
    \label{fig:ellipse2}
\end{figure}

\paragraph{Slopes, slack, and stretch.}
The \emph{slope} of a line segment $ab$ is $\slope(ab)=\frac{y(b)-y(a)}{x(b)-x(a)}$ if $x(a)\neq x(b)$, and $\slope(ab)=\infty$ if $x(a)=x(b)$.
For a line segment $ab$ in the plane, we denote by $\proj(ab)$ the orthogonal projection of $ab$ to the $x$-axis. Note that $|\proj(ab)|=|x(b)-x(a)|$. 
We define the \emph{slack} of $ab$ as $\slack(ab)=d(a,b)-|\proj(ab)|$. 
The Taylor series gives a relation between slopes and Euclidean distance. 

\begin{lemma}\label{lem:slack}
For any line segment $ab$ with $\slope(ab)\neq \infty$, we have 
\begin{align}
|\proj(ab)|\cdot \left(1+\frac13\cdot |\slope(ab)|^2\right) ~\leq~ 
& d(a,b)
~\leq~  |\proj(ab)|\cdot \left(1+\frac12\cdot |\slope(ab)|^2\right) \label{eq:1}\\
  \frac13\cdot |\slope(ab)|^2 ~\leq~ 
 & \frac{\slack(ab)}{|\proj(ab)|} 
   ~\leq~ \frac12\cdot |\slope(ab)|^2 \label{eq:2} .
\end{align}
\end{lemma}
\begin{proof}
By Pythagoras' theorem $d(a,b)=|\proj(ab)|\cdot \sqrt{1+|\slope(ab)|^2}$. The Taylor estimates of the function $f(x)=\sqrt{1+x^2}$ is $f(x)=1+\frac{x^2}{2}-\frac{x^4}{8}+O(x^5)$ near $x=0$. Therefore, $1+\frac{x^2}{3}\leq f(x) <1+\frac{x^2}{2}$ for $|x|\leq 1$. Substituting $x=|\slope(ab)|$ and $\slack(ab)=d(a,b)-|\proj(ab)|$ completes the proof.
\end{proof}

It is well known that an $x$-monotone path {(i.e., a path in which the $x$-coordinates of the points along the path are monotone increasing)} with edges of bounded slopes have small stretch. We include a proof for completeness. 

\begin{lemma}\label{lem:slope0}
Let $\pi=(v_0,v_1,\ldots, v_k)$ be an $x$-monotone polygonal path in $\mathbb{R}^2$ such that 
$|\slope(v_{i-1}v_i)|\leq \varrho$ for all $i=1,\ldots , k$.
Then $|\slope(v_0v_k)|\leq \varrho$; and
\[
    \frac{\wts(\pi)}{d(v_0,v_k)} \leq \sqrt{1+\varrho^2} < 1+\frac{\varrho^2}{2} .
\]
\end{lemma}
\begin{proof}
Then every edge $v_{i-1}v_i$ of $\pi$ satisfies $|\slope(v_{i-1}v_i)|\leq \varrho$. This implies $|y(v_i)-y(v_{i-1})|\leq \varrho \cdot |x(v_i)-x(v_{i-1})|$. Since the path $\pi$ is $x$-monotone, then we have 
\begin{align*}
    |y(v_k)-y(v_0)|
    & = \left| \sum_{i=1}^k y(v_{i})-y(v_{i-1}) \right|\\ 
    &\leq  \sum_{i=1}^k |y(v_i)-y(v_{i-1})| \\
    &\leq \sum_{i=1}^k \varrho \cdot |x(v_i)-x(v_{i-1})| \\
    & = \varrho\cdot |x(v_k)-x(v_0)|.
\end{align*}
This implies $|\slope(v_kv_0|\leq \varrho$. For the second claim, note that 
\begin{align*}
\wts(\pi) 
    &=\sum_{i=1}^k d(v_{i-1},v_i) 
    \leq \sqrt{1+\varrho^2} \cdot \sum_{i=1}^k |\proj(v_{i-1}v_i))|\\
    &= \sqrt{1+\varrho^2}\cdot |\proj(v_0v_k)| 
    \leq \sqrt{1+\varrho^2}\cdot d(v_0,v_k) 
    < \left(1+\frac{\varrho^2}{2}\right)\cdot d(v_0,v_k) .  \qedhere
\end{align*}
\end{proof}

\paragraph{Divide-and-Conquer for Minimum Spanning Trees.}
Let $P$ be a finite set in a metric space $(X,d)$. A ball of radius $r$ centered at $c\in X$ is denoted by $B(c,r)=\{x\in X: d(c,x)<r\}$. We use an easy observation that, under mild assumptions, the weight of the minimum spanning tree of $P$, denoted $\mst(P)$, is bounded below by the sum of weights of MSTs of subsets of $P$. 

A \emph{shallow cover} of $P$ is a collection of metric balls $\mathcal{C}=\{B(a_i,r_i): i\in I\}$ such that
\begin{itemize}
\item (\emph{cover}) $P\subset \bigcup_{i\in I} B(a_i,r_i)$; and 
\item (\emph{shallow}) for every $i\in I$, the ball $B(a_i,3r_i)$ intersects $O(1)$ balls in $\mathcal{C}$ .
\end{itemize}
With this notation, we prove the following. 

\begin{lemma}\label{lem:cover}
Let $P$ be a set of $n$ points in a metric space $(X,d)$, and let $\mathcal{C}=\{B_i: i\in I\}$ be a shallow cover of $P$. Then 
\[ 
    \sum_{i\in I} \wts(\mst(P\cap B_i))  
    \leq O(\wts(\mst(P))) .
\]
\end{lemma}
\begin{proof}
    Let $B_i=B(a_i,r_i)$ for all $i\in I$. We may assume w.l.o.g.\ that $P\cap B_i\neq \emptyset$ for all $i\in I$. 
    Consider the intersection graph $\mathcal{G}$ of the balls of triple radii, $\{B(a_i,3r_i): i\in I\}$: The nodes correspond to balls $B(a_i,3r_i)$, and the edges correspond to intersecting pairs of balls. Since  $\mathcal{C}$ is a shallow cover, then $\mathcal{G}$ has bounded degree, hence it admits a proper coloring with $O(1)$ colors. 
    
    Let $\{B(a_i,3r_i): i\in J\}$ be a color class of $\mathcal{G}$ (where $J\subset I$), and note that the balls of triple radii $\{B(a_i,3r_i): i\in J\}$ are pairwise disjoint. Let $P_J=\bigcup_{i\in J} P\cap B_i$. Clearly, we have $P_J\subseteq P$, and so $\wts(\mst(P_J))\leq \wts(\mst(P))$. For each $i\in J$, the maximum distance between any two points in $P\cap B_i$ is at most ${\rm diam}(P\cap B_i)\leq {\rm diam}(B_i)=2r_i$. Since the balls in $\{B(a_i,3r_i):i\in J\}$ are disjoint, the distance between any point in $P\cap B_i$ and a point in $P_J\setminus B_i$ is more than $2r_i$. Consequently, Kruskal's algorithm on $P_J$ constructs $\mst(P_i)$ before adding any edge between $P\cap B_i$ and $P_J\setminus B_i$. In particular, $\mst(P_J)$ contains $\mst(P\cap B_i)$ for all $i\in J$. This immediately implies 
  \[ 
  \sum_{i\in J} \wts(\mst(P\cap B_i))  
\leq O(\wts(\mst(P))) .
\]
   Summation over $O(1)$ color classes yields 
     \[ 
  \sum_{i\in I} \wts(\mst(P\cap B_i))  
\leq O(\wts(\mst(P))) ,
\]
as required.
\end{proof}

\section{Reduction to Net Points in a Trapezoid}
\label{sec:reduction}

In this section, we reduce the problem of constructing a bicriteria approximation for the minimum weight (Steiner) $(1+\eps)$-ST to the special case where all points, except the source $s$, lie in a trapezoid, and the point set is sparse (i.e., form a centered $\eps$-net, defined below). 

\subsection{Localization for a Tile}
\label{ssec:SteinerReduction}

Given a source $s\in \mathbb{R}^2$ and $\eps>0$, define a {\em tiling} of the plane into a set $\mathcal{T}$ of trapezoids as follows; refer to \Cref{fig:spiderweb}. Let $C$ be a circle of unit radius centered at $s$. 
Let $O_1$ be a regular $k$-polygon with inscribed circle $C$, where $k$ is the minimum integer such that the side length of $O_1$ is less than $\sqrt{\eps}$. For all integers $i\in \mathbb{Z}$, let $O_i=2^i\cdot O_1$, that is, a scaled copy of $O_1$, centered at $s$.
Finally, add $k$ rays emanating from $s$ that pass through the vertices of the polygons $O_i$, $i\in \mathbb{Z}$. Polygons $O_i$, $i\in \mathbb{Z}$, and the $k$ rays subdivide the plane into a set $\mathcal{T}$ of trapezoids: Each trapezoid $\tau\in \mathcal{T}$ lies between two consecutive polygons $O_i$ and $O_{i+1}$, and two consecutive rays. 

\begin{figure}[htbp]
    \centering
    \includegraphics[width=.5\columnwidth]{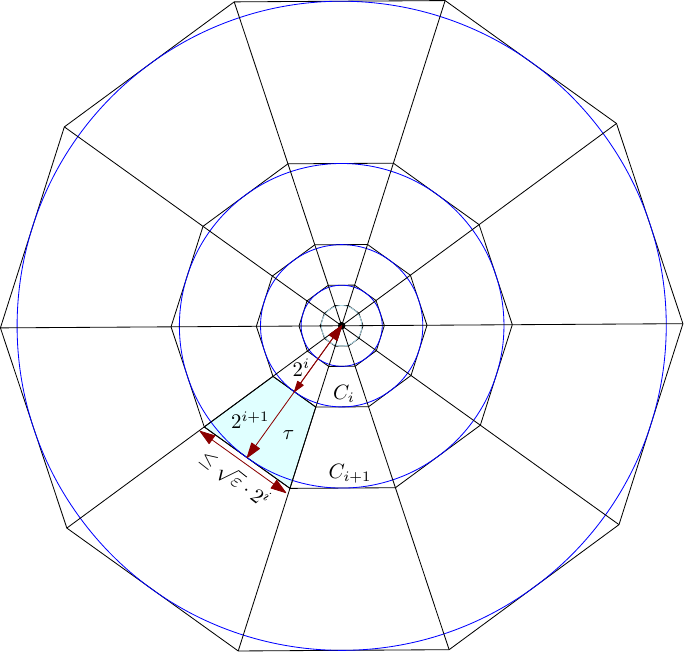}
    \caption{A section of the trapezoid tiling of the plane.}
    \label{fig:spiderweb}
\end{figure}

For each tile $\tau\in \mathcal{T}$, we define a region $R(\tau)\subset \mathbb{R}$ as follows. Recall that for any $p\in \mathbb{R}^2$, $\mathcal{E}_{ps}$ denotes the ellipse with foci $p$ and $s$ and major axis $(1+\eps)\cdot d(p,s)$. Let $R(\tau)=\bigcup_{p\in \tau} \mathcal{E}_{ps}$; see \Cref{fig:region}.   
Note that $s\in R(\tau)$ for all $\tau\in \mathcal{T}$. We also define a smaller region that includes $\tau$ but excludes a neighborhood of $s$: Let $L_\tau$ be the orthogonal bisector of the line segment between $s$ and the point in $\tau$ that is closest to $s$; and let $L_\tau^-$ be the halfplane bounded by $L_\tau$ such that $\tau\subset L_\tau^-$. 
Now let $R^-(\tau)=R(\tau)\cap L_\tau^-$. 

\begin{figure}[htbp]
    \centering
    \includegraphics[width=.9\columnwidth]{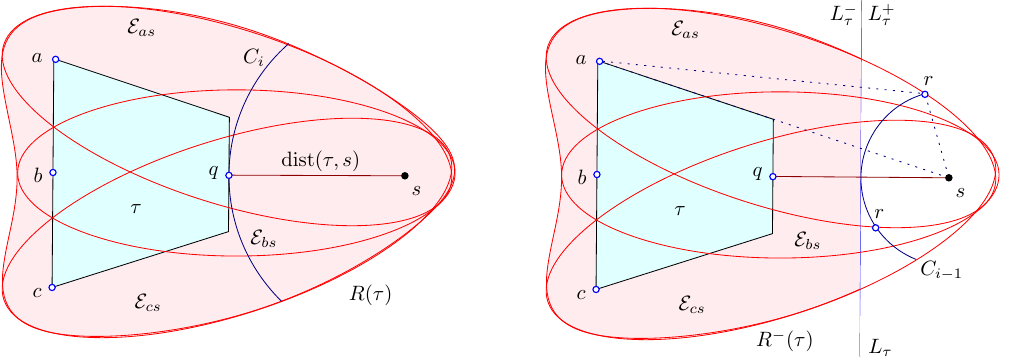}
    \caption{Left: region $R(\tau)$ is the union of ellipses $\mathcal{E}_{ps}$ for all $p\in \tau$ (including $p\in \{a,b,c\}$). Right: Region $R^-(\tau)$ and the line $L_\tau$.}
    \label{fig:region}
\end{figure}

\begin{lemma}\label{lem:tiles}
    For every tile $\tau\in \mathcal{T}$, the region $R^-(\tau)$ intersects $O(1)$ tiles in $\mathcal{T}$.
\end{lemma}
\begin{proof}
The tile $\tau$ lies between the polygons $O_i$ and $O_{i+1}$ for some $i\in \mathbb{Z}$.
The point $q\in \tau$ closest to $s$ is on the circle $C_i$ of radius $2^i$ centered at $s$; see \Cref{fig:region}. Then $L_\tau$ is the orthogonal bisector of $qs$, and so $L_\tau$ is tangent to the circle $C_{i-1}$ of radius $2^{i-1}$ centered at $s$. This implies that $R^-(\tau)$ lies in the exterior of $C_{i-1}$.

For every point $a\in \tau$, we have $d(a,s)\leq 2^i \sqrt{4+\eps/4} < (1+\eps)2^{i+1}$. This gives $d(p,s)\leq (1+\eps)d(a,s)< (1+\eps)^2 2^{i+1}< (1+3\eps)2^{i+1}$ for every point $p \in R^-(\tau)$. Therefore, the region $R^-(\tau)$ lies in the annulus $A$ between two concentric circles of radius $2^{i-1}$ and $(1+3\eps)2^{i+1}$, centered at $s$. 

Next, we show that $R^-(\tau)$ is contained in a cone with apex $s$ and aperture $O(\sqrt{\eps})$. For an arbitrary point $a\in \tau$, consider the ellipse $\mathcal{E}_{as}$, and let $r\in \mathcal{E}_{ab}\cap C_{i-1}$. Then $d(a,r)+d(r,s)=(1+\eps)\cdot d(a,s)$. In the triangle $\Delta(ars)$, the law of cosines yields
\begin{align*}
    \cos (\angle asr)
    &= \frac{(d(a,s))^2+(d(r,s))^2-(d(a,r))^2}{2\cdot d(a,s)\cdot d(r,s)}\\
    &= \frac{(d(a,s))^2+(d(r,s))^2-((1+\eps)\cdot d(a,s)-d(r,s))^2}{2\cdot d(a,s)\cdot d(r,s)}\\
    &=\frac{2(1+\eps)\cdot d(a,s)\cdot d(r,s)- (2\eps+\eps^2)\cdot d(a,s)}{2\cdot d(a,s)\cdot d(r,s)}\\
    &= 1+\eps - \frac{\eps(2+\eps)\cdot d(a,s)}{2\cdot d(r,s)}\\
    & \geq 1+\eps\left(1 -\frac{(2+\eps)(1+\eps)2^{i+1}}{2\cdot 2^{i-1}}\right)\\
    &\geq 1-O(\eps).
\end{align*}
The Taylor estimate $\cos x\geq 1-\frac{x^2}{2}$ implies that
$(\angle asr)^2/2\leq O(\eps)$, and so $\angle asr \leq O(\sqrt{\eps})$.

By the construction of the tile $\tau$, we have $\angle asq\leq O(\sqrt{\eps})$ for every point $a\in \tau$. Consequently, $\angle qsr \leq \angle qsa+ \angle ars \leq O(\sqrt{\eps})+O(\sqrt{\eps})=O(\sqrt{\eps})$. 
We conclude that the region $R^-(\tau)$ is contained in a cone $W$ with apex $s$ and aperture $O(\sqrt{\eps})$.

We have shown that $R^-(\tau)\subset A$ and $R^-(\tau)\subset W$, for the annulus $A$ and the cone $W$ defined above, and so $R^-(\tau)\subset A\cap W$. It is clear from the definition of the tiling $\mathcal{T}$ that $A\cap W$ intersects $O(1)$ tiles in $\mathcal{T}$, therefore $R^-(\tau)$ also intersects $O(1)$ tiles in $\mathcal{T}$.
\end{proof}

\subsection{Reduction to Points in a Trapezoid}
\label{ssec:nonSteinerReduction}

Let $P$ be a set of $n$ points in the plane, and let $\tau_1,\ldots , \tau_m\in \mathcal{T}$ be the tiles that contain at least one point in $P$, and let $P_i:=P\cap \tau_i$. 
For a tile $\tau_i\in \mathcal{T}$, we define a \EMPH{tile-restricted $(1+\eps)$-ST} (for short, \EMPH{$(1+\eps)$-tST}) as a tree $T=(V,E)$ such that
\begin{itemize}\itemsep 0pt
   \item $P_i\cup \{s\}\subseteq V(T)\subseteq P$, and 
    \item $T$ contains a $ps$-path of weight at most $(1+\eps)\cdot d(p,s)$ for every $p\in P_i$.
\end{itemize}
In other words, a tile-restricted \emph{$(1+\eps)$-ST} is a Steiner $(1+\eps)$-ST for $P_i\cup \{s\}$ rooted at $s$, where all Steiner points are in $P$.

In the remainder of this section, we compare the weights of a $(1+\eps)$-ST for $P$, rooted at $s$, and the total weights of tile-restricted $(1+\eps)$-STs for $P_i\cup \{s\}$ for $i=1,\ldots , m$. The arguments in this section directly extend to Steiner $(1+\eps)$-STs, the only difference is that a Steiner $(1+\eps)$-ST can use any Steiner point in the plane, while a tile-restricted $(1+\eps)$-ST can only use Steiner points from the input set $P$. 

\begin{lemma}\label{lem:nonSteiner-converse}
Let $H$ be a $(1+\eps)$-ST (resp., Steiner $(1+\eps)$-ST) for a point set $P\cup \{s\}$ in the plane. Then for every $i=1,\ldots , m$, there is a polynomial-time construction of a tile-restricted $(1+\eps)$-ST $H_i$ (resp., Steiner $(1+\eps)$-ST) for the point set $P_i\cup \{s\}$, where $P_i=P\cap \tau_i$, such that 
\begin{equation}\label{eq:weightH+}
    \sum_{i=1}^m \wts(H_i)\leq O(\wts(H)) .
\end{equation}
\end{lemma}
\begin{proof}
Let $H$ be a $(1+\eps)$-ST (resp., Steiner $(1+\eps)$-ST) for $P\cup \{s\}$ rooted at $s$. For every $i=1,\ldots ,m$, we first construct a graph $H_i$, then show that $H_i$ is a tile-restricted $(1+\eps)$-ST for $P_i\cup \{s\}$, and finally establish \Cref{eq:weightH+}.

Consider a tile $\tau_i\in \mathcal{T}$; we write $\tau=\tau_i$ for simplicity. Similarly to \Cref{ssec:SteinerReduction}, let $L_\tau$ be the orthogonal bisector of the line segment between $s$ and the point in $\tau$ that is closest to $s$; and let $L_\tau^-$ and $L_\tau^+$ be the two halfplanes determined by $L_\tau$ such that $\tau\subset L_\tau^-$ and $s\in L_\tau^+$; see \Cref{fig:paths}.

  For every $i$ and $p\in P_i$, the (Steiner) $(1+\eps)$-ST $H$ contains a $ps$-path $\pi_{ps}$ of weight at most $(1+\eps)\cdot d(p,s)$. Note that $\pi_{ps}\subset \mathcal{E}_{ps}\subset R(\tau_i)$, and so all vertices of $\pi_{ps}$ are in $P$. 
  Let $H_i=\bigcup_{p\in P_i} \pi_{ps}$. Since $H$ is a tree and all paths $\pi_{ps}$, $p\in P_i$, end at $s$, then $H_i$ is a tree rooted at $s$.

\begin{figure}[htbp]
    \centering
    \includegraphics[width=.85\columnwidth]{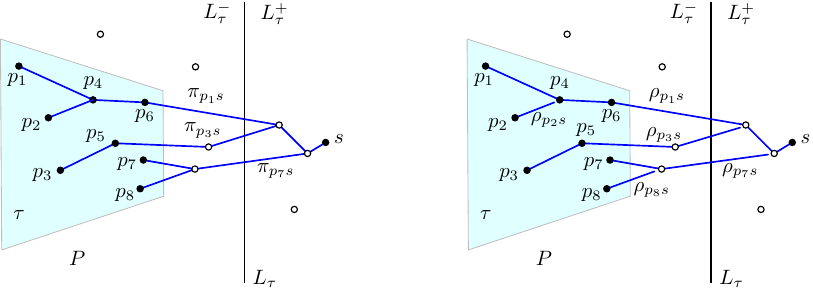}
    \caption{Left: Shortest paths $\pi_{ps}$ in $H$ from points $p\in P\cap \tau$ to $s$. 
    Right: The decomposition of the tree $H_i=\bigcup_{p\in P_i}\pi_{ps}$ into interior-disjoint paths $\rho_{ps}$, $p\in P_i$. }
    \label{fig:paths}
\end{figure}

\paragraph{Root stretch analysis.} For every $p\in P_i$, $\wts(\pi_{ps})\leq (1+\eps)\cdot d(p,s)$ since $H$ is a (Steiner) $(1+\eps)$-ST for $P\cup \{s\}$ rooted at $s$. 
Since $\pi_{ps}\subset H_i$ and all its vertices are in $P$, then $H_i$ is a tile-restricted $(1+\eps)$-ST (resp., $(1+\eps)$-ST) for $P_i\cup \{s\}$ rooted at $s$.

\paragraph{Weight analysis.}
Assume w.l.o.g.\ that the symmetry axis of the tile $\tau_i$ is the $x$-axis, 
$0\leq x(p)\leq 1$ for all $p\in P_i$ and $s=(2,0)$. 
Then $L_\tau$ is the vertical line $L_\tau:x=\frac32$. For two points $p,p'\in P_i$, the shortest paths $\pi_{ps}$ and $\pi_{p's}$ may overlap, but their intersection $\pi_{ps}\cap \pi_{ps}$ is connected and incident to $s$. 
We define interior-disjoint paths as follows; refer to \Cref{fig:paths}(right). 
Order the points in $P_i$ arbitrarily as $p_1,\ldots , p_{n_i}$. For each $j=1,\ldots ,n_i$, let $\rho_{p_js}$ be the portion of the path $\pi_{p_js}$ from $p_j$ to the first intersection with any previous path $\pi_{p_\ell s}$, $\ell<j$. Observe that $\bigcup_{j=1}^{n_i} \pi_{p_j s}=\bigcup_{j=1}^{n_i} \rho_{p_j s}$, which gives 
\begin{equation}\label{eq:decompose}
   \wts(H_i) 
    = \sum_{j=1}^{n_i}  \wts(\rho_{p_j s}).
\end{equation}

Consider a path $\rho_{ps}$ for $p\in P_i$. 
\begin{itemize}
\item If the last vertex of $\rho_{ps}$ is in the halfplane $L_\tau^+$ (possibly $s$).
Then $x(p)\leq 1$ and $L_\tau:x=\frac32$ yield $\wts(\rho_{ps}\cap L_\tau^-)\geq \frac12$.
On the other hand $\rho_{ps}\subset \pi_{ps}$ implies that 
$\wts(\rho_{ps})\leq (1+\eps)\cdot \wts(p,s)\leq 4$.
Therefore, $\wts(\rho_{ps}\cap L_\tau^-)\geq \Omega(\wts(\rho_{ps}))$. 
\item If all vertices of $\rho_{ps}$ are in the halfplane $L_\tau^-$,
then we have $\wts(\rho_{ps}\cap L_\tau^-)=\wts(\rho_{ps})$.
\end{itemize}
In both cases, we have $\wts(\rho_{ps}\cap L_\tau^-)\geq \Omega(\wts(\rho_{ps}))$. 
Note also that for every $p\in P_i$, $\wts(\pi_{ps})\leq (1+\eps)\cdot d(p,s)$ implies that 
$\rho_{ps}\subset \pi_{ps}\subset \mathcal{E}_{ps}\subset R(\tau_i)$, and so $\rho_{ps}\cap L_{\tau_i}^-\subset \pi_{ps}\cap L_{\tau_i}^-\subset R^-(\tau_i)$. We combine these observations:
\begin{equation}\label{eq:shortcut2}
  O\left( \sum_{j=1}^{n_i} \wts(\rho_{p_j s})\right)
    \leq O\left( \sum_{j=1}^{n_i}  \wts(\rho_{p_j s}\cap L_{\tau_i}^-)\right) 
    = O\left( \sum_{j=1}^{n_i} \wts(\rho_{p_j s}\cap R^-_{\tau_i})\right)
    \leq O(\wts(H\cap R_{\tau_i}^-)).
\end{equation}
By \Cref{lem:tiles}, each point in $H$ is contained in $H\cap R_\tau^-$ for $O(1)$ tiles $\tau\in \mathcal{T}$. This implies that 
\begin{equation}\label{eq:shortcut3}
    \sum_{\tau\in \mathcal{T}} \wts(H\cap R^-(\tau))\leq O(\wts(H)) .
\end{equation}
The combination of \Cref{eq:decompose,eq:shortcut2,eq:shortcut3} yields $\sum_{i=1}^m \wts(H_i) \leq O(\wts(H))$, as required. 
\end{proof}

We can now state our reduction to a single tile. The same argument works for both the Steiner and the non-Steiner settings. 

\begin{theorem}\label{thm:nonSteiner-reduction}
	We are given a set $P\subset \mathbb{R}^2$ of $n$ points, a source $s\in \mathbb{R}^2$, a parameter $\eps>0$, and two real functions $f(.)$ and $g(.)$.
    Let $\tau_1,\ldots ,\tau_m$ be the set of tiles in $\mathcal{T}$, where $P_i:=P\cap \tau_i\neq \emptyset$. Let $G_i$ be a tile-restricted $(1+f(\eps)\cdot \eps)$-ST 
    (resp., Steiner $(1+f(\eps)\cdot \eps)$-ST) for $P_i\cup \{s\}$ of weight $\wts(G_i)\leq O(g(\eps)\cdot \opt_{\eps,i})$, where $\opt_{\eps,i}$ is the minimum weight of a $(1+\eps)$-ST for $P_i\cup \{s\}$, for all $i=1,\ldots ,m$.
    
    Then a shortest-path tree of the graph $G:=\bigcup_{i=1}^m G_i$, rooted at $s$, is a $(1+f(\eps)\cdot \eps)$-ST (resp., Steiner $(1+f(\eps)\cdot \eps)$-ST) for $P\cup \{s\}$ and  $\wts(G)\leq O(g(\eps)\cdot \opt_{\eps})$, 
    where $\opt_{\eps}$ is the minimum weight of a $(1+\eps)$-ST for $P\cup \{s\}$.
\end{theorem}
\begin{proof}
    The root stretch analysis of $G=\bigcup_{i=1}^m G_i$ is immediate: For every point $p\in P$, we have $p\in P_i$ for some $i\in \{1,\ldots , m\}$. As $G_i$ is a  
    tile-restricted $(1+f(\eps)\cdot \eps)$-ST for $P_i\cup \{s\}$ (resp., Steiner $(1+f(\eps)\cdot \eps)$-ST), it contains a $ps$-path of weight at most $(1+f(\eps)\cdot \eps)\cdot d(p,s)$, where all vertices of the path are in $P$ (resp., Steiner points), and so $G=\bigcup_{i=1}^m G_i$ also contains this path.

    For the weight analysis, we have $\wts(G)\leq \sum_{i=1}^m \wts(G_i)\leq O(g(\eps))\sum_{i=1}^m \opt_{\eps,i}$. Let $H$ be a minimum-weight $(1+\eps)$-ST (resp., Steiner $(1+\eps)$-ST) for $P\cup \{s\}$, that is,  $\wts(H)=\opt_{\eps}$. Then \Cref{lem:nonSteiner-converse} implies
    \begin{equation}\label{eq:reduction}
    \sum_{i=1}^m \opt_{\eps,i}\leq O(\opt_{\eps}),
    \end{equation}
    which completes the proof.  
\end{proof}

\subsection{Reduction to a Centered Net}
\label{ssec:net}

For a set $P$ in a metric space $(X,d)$ and a parameter $\eps>0$, an \EMPH{$\eps$-net} is a subset $N\subset P$ such that for all $a,b\in N$, $a\neq b$, we have $d(a,b)>\eps$, and for every $p\in P$, there exists a point $a\in N$ such that $d(p,a)\le\eps$. 

For the purpose of a $(1+\eps)$-ST, with a source (or center) $s\in P$, we define a variation of $\eps$-nets where the density of the net depends on the distance from  $s$: A \EMPH{centered $\eps$-net} (for short, \EMPH{$\eps$-cnet}) with center $s$ is a subset $N\subset P$ such that 
\begin{itemize}
\item for all $a,b\in N$, $a\neq b$, we have $d(a,b)> \eps\cdot \min\{d(a,s),d(b,s)\}$ and 
\item for every $p \in P$, there exists a point $a\in N$ such that $d(p,a)\leq \eps \cdot d(p,s)$.
\end{itemize}

Given a point set $P\subset \mathbb{R}^2$, including $s\in P$, a parameter $\eps>0$ and a subset $P$, we also define a \EMPH{net-restricted $(1+\eps)$-ST for $N\subset P$} as a tree $T=(V,E)$ such that
\begin{itemize}\itemsep 0pt
   \item $N \cup \{s\}\subseteq V(T)\subseteq P$, and 
    \item $T$ contains a $ps$-path of weight at most $(1+\eps)\cdot d(p,s)$ for every $p\in N$.
\end{itemize}
This means that $T$ is a Steiner $(1+\eps)$-ST rooted at $s$ for the $\eps$-cnet $N$, where all Steiner points are in $P$. 
In this subsection, we show that it suffices to find a net-restricted (resp., Steiner) $(1+\eps)$-ST for an $\eps$-cnet with center $s$. We state the result (\Cref{lem:net}) for both Steiner and non-Steiner settings. 

\begin{lemma}\label{lem:net}
    Let $P$ be a set of $n$ points in the plane, $s\in P$, $\eps\in (0,\frac19)$, and $N\subset P$ an $\eps$-cnet for $P$. Let $f(.)$ and $g(.)$ be real functions such that $f(x)\geq x$ and $g(x)\geq x$ for all $x\geq 0$. 
    Assume that in $O(|N|\log |N|)$ time, we can compute a net-restricted (resp., Steiner) $(1+f(\eps)\cdot\eps)$-ST $T_N$ for $N$ of weight $O(g(\eps))\cdot \opt_{\eps}(N)$, where $\opt_{\eps}(N)$ is the minimum weight of a net-restricted (resp., Steiner) $(1+\eps)$-ST for $N$. 

    Then in $O(n\log n)$ time, we can compute a (Steiner) $(1+O(f(\eps)\cdot \eps))$-ST $T_P$ for $P$ of weight $O(g(\eps))\cdot \opt_{\eps}(P)$, where $\opt_{\eps}(P)$ is the minimum weight of a (Steiner) $(1+\eps)$-ST for $P$. 
\end{lemma}
\begin{proof}
An $\eps$-cnet is associated with a partition (\emph{clustering}) of $P$, where each cluster corresponds to points in a neighborhood of a net point. Specifically, 
note that $P\subset \bigcup_{a\in N} B(a,\eps(1+2\eps)\cdot d(a,s))$. Indeed, for every $p\in P$, we have $d(p,a)<\eps \cdot d(p,s)$ for some point $a\in N$. The triangle inequality gives $d(p,s)\leq d(p,a)+d(a,s)< \eps\cdot d(p,s)+d(a,s)$. Thus, we have $d(p,s) <d(a,s)/(1-\eps) < (1+2\eps)\cdot d(a,s)$ for all $0<\eps<\frac12$. This yields $d(p,a) < \eps \cdot d(p,s) < \eps (1+2\eps)\cdot d(a,s)$, and so $p\in B(a,\eps(1+2\eps)\cdot d(a,s))$. 

We can define a \EMPH{cluster} of $a\in N$ such that $C_a=P\cap B(a,\eps(1+2\eps)\cdot d(a,s))$. Specifically, order the points in $N$ arbitrarily. For every $a\in N$ in this order, let $C_a$ be the set of points in $P\cap B(a,\eps(1+2\eps)\cdot d(a,s))$ that have not been included in previous clusters. Since $1+2\eps<\frac{11}{9}<\frac43$ for $\eps\in (0,\frac19)$, we use the covering $P\subset \bigcup_{a\in N} B(a,\frac{4\eps}{3}\cdot d(a,s))$ in the sequel.

For each cluster $C_a$, $a\in N$, let $G_a$ be a 
$2$-spanner of weight $O(\wts(\mst(C_a))$, which can be computed in $O(|C_a|\log |C_a|)$ time~\cite{KanjPX10}. 
Let $G$ be the union of the (Steiner) $(1+f(\eps) \cdot \eps)$-ST $T_N$ and the 2-spanners $G_a$ for all $a\in N$. Clearly, $G$ is a connected (Steiner) graph on $P$:  Let $T_P$ be the shortest-path tree $\mathsf{SPT}(G)$ of $G$ rooted at $s$. 

\paragraph{Root stretch analysis.} For a point $p\in P$, assume that $p\in C_a$ for some $a\in N$. Then $G$ contains a $ps$-path $\pi_{ps}$ as a concatenation of a shortest path $\pi_{pa}$ in the 2-spanner $G_a$ and an $as$-path $\pi_{as}$ in $T_N$. 
Then we have 
\begin{align*}
    \wts(\pi_{ps}) 
    &=\wts(\pi_{pa})+\wts(\pi_{as})\\
    &\leq 2\cdot d(p,a)+(1+f(\eps)\cdot \eps)\cdot d(a,s)\\
    &\leq 2\cdot d(p,a)+(1+f(\eps)\cdot\eps)\cdot \big(d(a,p)+d(p,s)\big)\\
    &= (3+f(\eps) \cdot \eps)\cdot d(p,a)+(1+f(\eps)\cdot\eps)\cdot d(p,s)\\
    &< \eps (3+f(\eps)\cdot\eps)\cdot d(p,s)+(1+f(\eps)\cdot\eps)\cdot d(p,s)\\
    &=(1+3\eps+f(\eps)\cdot \eps+f(\eps)\cdot \eps^2)\cdot d(p,s)\\
    &<(1+5f(\eps)\cdot \eps)\cdot d(p,s).
\end{align*}
Consequently, the shortest $ps$-path in $T_P$ has weight at most $(1+O(f(\eps)\cdot \eps))\cdot d(p,s)$.

\paragraph{Weight analysis.} Since $N\subset P$, we have $\wts(\opt_{\eps}(N))\leq \wts(\opt_\eps(P))$, where $\opt_{\eps}(P)$ is the minimum weight of a (Steiner) $(1+\eps)$-ST for $P$. We can bound the total weight of the 2-spanners $G_a$, $a\in N$, using \Cref{lem:cover}. It is enough to show that the collection of balls 
\[
    \mathcal{C}=\left\{B\left(a, \frac{4\eps}{3}\cdot d(a,s)\right): a\in N\right\}
\]
forms a shallow cover of $P$. This can be done with a standard volume argument. Since $N$ is an $\eps$-cnet, none of the balls in $\mathcal{C}'=\{B(a, \eps\cdot d(a,s)): a\in N\}$ contains the center of any other ball, This implies that balls of half-radii in $\mathcal{C}''=\{B(a, \frac{\eps}{2}\cdot d(a,s)): a\in N\}$ are pairwise disjoint. 

Assume that $a,b\in N$, and the balls $B(a,4\eps\cdot d(a,s))$ and $B(b,4\eps\cdot d(b,s))$ intersect in some point $x\in \mathbb{R}^2$. The triangle inequality yields $d(x,s)\geq d(a,s)-d(s,x)\geq (1-4\eps)d(a,s)>\frac12\cdot d(a,s)$ 
for $\eps<\frac19$, hence $d(a,x)< 4\eps \cdot d(a,s)<8\eps\cdot d(x,s)$. Similarly, $d(b,x)< 8\eps\cdot d(x,s)$, and hence $d(a,b)\leq d(a,x)+d(x,b)< 16\eps\cdot d(x,s)$.
The triangle inequality also gives $d(x,s)\leq d(x,a)+d(a,s)\leq (4\eps+1)d(a,s)<  \frac32\cdot d(a,s)$ for $\eps<\frac19$;
and symmetrically $d(x,s)< \frac32\cdot d(b,s)$.
Therefore, we obtain $d(a,b)\leq 16\eps\cdot d(x,s)< 24\eps\cdot d(a,s)$, and 
$d(b,s)\leq d(b,a)+d(a,s)\leq (1+24\eps)\cdot d(a,s)< 4\cdot d(a,s)$ for $\eps\in (0,\frac19)$.

Recall that the balls in $\mathcal{C}''=\{B(a, \frac{\eps}{2}\cdot d(a,s)): a\in N\}$ are pairwise disjoint. From the previous paragraph, we can conclude that if $a,b\in N$ and the balls $B(a,4\eps\cdot d(a,s))$ and $B(b,4\eps\cdot d(b,s))$ intersect, then
\begin{align*}
B\left(b, \frac{\eps}{2} \cdot d(b,s)\right) & \subset B(b,2\eps \cdot d(a,s)),\\
B\left(b, \frac{\eps}{2} \cdot d(b,s)\right)&\subset  B\left(a,d(a,b)+\frac{\eps}{2} \cdot d(b,s)\right) \subset B(a,26\eps\cdot d(x,s)).
\end{align*}
By volume argument, the ball of radius $26\eps\cdot d(a,s)$ centered at $a$ contains at most $(2\cdot 26)^2=O(1)$ disjoint balls of radius $\frac{\eps}{2}\cdot d(a,s)$, which are pairwise disjoint. Consequently, for every $a\in N$, there are $O(1)$ points $b\in N$ such that $B(a,6\eps\cdot d(a,s))$ intersects $B(a,6\eps\cdot d(b,s))$.

\paragraph{Running time.} 
The $n$ points can greedily be partitioned into clusters in $O(n\log n)$ time using the ice cream scoop algorithm supported by a point location data structure~\cite{EdelsbrunnerGS86, SarnakT86} (see also \cite{IaconoM12}).
For each cluster $C_a$, $a\in N$, a $2$-spanner $G_a$ of weight $O(\wts(\mst(C_a))$, which can be computed in $O(|C_a|\log |C_a|)$ time~\cite{KanjPX10}. Hence all 2-spanners $G_a$, $a\in N$, can be computed in $\sum_{a\in N} O(|C_a| \log |C_a|)\leq  \left(\sum_{a\in N} |C_a|\right)\cdot O(\log n)=O(n\log n)$ time. Consequently, we can compute $G$ in $O(n\log n)$ time. Note that $G$ has $O(n)$ vertices and edges, so $\mathsf{SPT}(G)$ can be computed in $O(n\log n)$ time. The overall running time is $O(n\log n)$.
\end{proof}

\subsection{Reduction to Ellipses with Parallel Major Axes}
\label{ssec:ellipses}

For two points $a,b\in \mathbb{R}^2$ and a parameter $\eps>0$, let $\mathcal{E}_{ab,\eps}$ denote the ellipse with foci $a$ and $b$ and major axis $(1+\eps)\cdot d(a,b)$, that is,
\[
    \mathcal{E}_{ab,\eps}=\{p\in \mathbb{R}^2: d(a,p)+d(p,b)\leq (1+\eps)\cdot d(a,b)\} .
\]
We show that if $ps$ is a segment with $|\slope(ps)|\leq \sqrt{\eps}$, then the ellipse $\mathcal{E}_{ps,\eps}$ can be sandwiched between two ellipses of similar size that have a focus at $p$ and a horizontal major axis:
\begin{lemma}\label{lem:approxellipse}
    If $p,s\in \mathbb{R}^2$ and $|\slope(ps)|\leq \sqrt{\eps}$, then there exist points $a,b\in \mathbb{R}^2$ such that $pa$ and $pb$ are horizontal line segments, $d(p,b)\leq 2\cdot d(p,s)\leq 4\cdot d(p,a)$, and 
    \[
    \mathcal{E}_{pa,\eps/2} \subset \mathcal{E}_{ps,\eps} \subset \mathcal{E}_{pb,2\eps} .
    \]
Furthermore, for every vertical line $L$ that intersects the segment $pa$ left of the midpoint of $pa$, we have 
\begin{equation}\label{eq:widths}
    \wts(L\cap \mathcal{E}_{pa,\eps/2})
        \leq \wts(L\cap \mathcal{E}_{pb,2\eps}) 
        \leq 32\cdot \wts(L\cap \mathcal{E}_{pa,\eps/2}).
\end{equation}
\end{lemma}
\begin{figure}[htbp]
    \centering
    \includegraphics[width=.75\columnwidth]{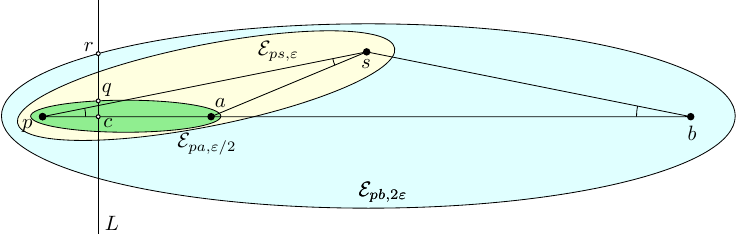}
    \caption{Ellipses $\mathcal{E}_{pa,\eps/2}$, $\mathcal{E}_{ps,\eps}$ and $\mathcal{E}_{pb,2\eps}$. A vertical line $L$ that intersects segment $pa$ left of its center, and crosses the boundary of $\mathcal{E}_{pa,\eps/2}$ and $\mathcal{E}_{pb,2\eps}$ at points $q$ and $r$, respectively.}
    \label{fig:ellipses}
\end{figure}
\begin{proof}
We choose points $a$ and $b$ on the horizontal line passing through $p$ that satisfy the equations: $d(p,a)=d(a,s)$ and $d(p,s)=d(s,b)$; see \Cref{fig:ellipses}. Note that the triangles $\Delta(aps)$ and $\Delta(bsp)$ are similar and isosceles, where $\angle aps =\angle bps = \tan^{-1} (\slope(ps))$.

By the triangle inequality, every point $q\in \mathcal{E}_{pa,\eps/2}$
satisfies 
\begin{align*}    d(p,q)+d(q,s) 
    & \leq d(p,q)+ d(q,a)+ d(a,s) \\
    &\leq  \left(1+\frac{\eps}{2}\right)\cdot d(p,a) + d(p,a) \\
    &= \left(2+\frac{\eps}{2}\right)\cdot d(p,a) \\
    &= \left(2+\frac{\eps}{2}\right)\cdot \frac{d(p,s)}{2\cos\angle bps}\\
    &\leq \frac{2+\eps/2}{2}\cdot \sqrt{1+|\slope(ps)|^2}\cdot d(p,s)\\
    &\leq  \left(1+\frac{\eps}{4}\right)\sqrt{1+\eps} \cdot d(p,s)\\
    &\leq  \left(1+\frac{\eps}{4}\right)\left(1+\frac{\eps}{2}\right) \cdot d(p,s)\\
    & <(1+\eps) \cdot d(p,s),
\end{align*}
where the last inequality holds for $\eps < 2$, which proves $\mathcal{E}_{pa,\eps/2} \subset \mathcal{E}_{ps,\eps}$.

Similarly, every point $q\in \mathcal{E}_{ps,\eps}$ satisfies 
\begin{align*}    d(p,q)+d(q,b) 
    & \leq d(p,q)+ d(q,s)+ d(s,b) \\
    &\leq  (1+\eps)\cdot d(p,s) + d(s,b) \\
    &=(2+\eps)\cdot d(p,s) \\
    &= (2+\eps)\cdot d(p,b)/ (2\cos\angle bps) \\ 
    &\leq \frac{2+\eps}{2}\cdot \sqrt{1+|\slope(ps)|^2}\cdot d(p,b)\\
    &\leq  \left(1+\frac{\eps}{2}\right)\sqrt{1+\eps} \cdot d(p,b)\\
    &\leq  \left(1+\frac{\eps}{2}\right)^2 \cdot d(p,b)\\
    & <(1+2\eps) \cdot d(p,b) ,
\end{align*}
which proves $\mathcal{E}_{ps,\eps} \subset \mathcal{E}_{pb,2\eps}$.

We now prove the chain of inequalities \eqref{eq:widths}. It is clear that   $\mathcal{E}_{pa,\eps/2} \subset \mathcal{E}_{pb,2\eps}$ implies $\wts(L\cap \mathcal{E}_{pa,\eps/2})\leq  \wts(L\cap \mathcal{E}_{pb,2\eps})$ for every line $L$. 
It remains to show that $\wts(L\cap \mathcal{E}_{pb,2\eps})\leq 32\cdot \wts(L\cap \mathcal{E}_{pb,2\eps})$ for a vertical line $L$ that intersects the line segment $pa$ left of the midpoint of $pa$. Let $c=L\cap pa$, and let $q,r\in L$ be the points in the upper halfplane such that $q$ is on the boundary of $\mathcal{E}_{pa,\eps/2}$ and $r$ is on the boundary of $\mathcal{E}_{pb,2\eps}$; see \Cref{fig:ellipses}.
Since $q\in \partial \mathcal{E}_{pa,\eps/2}$, then $\slack(pq)+\slack(qa)=\frac{\eps}{2}\cdot d(p,a)$. Since $L$ is left of the center of the ellipse $\mathcal{E}_{pa,\eps/2}$,
then $\slack(pq)\geq \slack(qa)$. Combined with $d(p,b)\leq 4\cdot d(p,a)$, we obtain  
\begin{equation}\label{eq:apx-ellipse}
        \slack(pq)\geq \frac{\slack(pq)+\slack(qa)}{2}
                \geq \frac{\eps}{4}\cdot d(p,a)
                \geq \frac{\eps}{16}\cdot d(p,b).
\end{equation}
Since $r\in \partial \mathcal{E}_{pb,2\eps}$, then  $\slack(pr)+\slack(pb)=2\eps\cdot d(p,b)$, and in particular $\slack(pr)\leq 2\eps\cdot d(p,b)$. We conclude that 
\[
    \slack(pr)\leq 32\cdot \slack(pq).
\]
Note that the $x$-projection of both segments $pq$ and $pq$ is the same segment $pc$. We can recover the length of the vertical segments $cq$ and $cr$ using Pythagoras theorem in the right triangles $\Delta(pcq)$ and $\Delta(pcr)$, respectively:
\begin{align*}
    \wts(L\cap \mathcal{E}_{pb,2\eps})
    &=2\cdot d(c,r) \\
    &=2\cdot \sqrt{\big((d(p,c))^2+(d(c,r))^2\big) -(d(p,c))^2 } \\
    &=2\cdot \sqrt{(d(p,r))^2 -(d(p,c))^2 }\\
    &=2\cdot \sqrt{(\slack(pr)+d(p,c))^2 -(d(p,c))^2 }\\
    &=2\cdot \sqrt{(\slack(pr))^2+2\cdot \slack(pr)\cdot d(p,c))^2}\\
    &\leq 2\cdot \sqrt{(32\cdot \slack(pq))^2+2\cdot 32\cdot \slack(pq)\cdot d(p,c))^2}\\
    &\leq 64\cdot \sqrt{\slack(pq))^2+2\cdot \slack(pq)\cdot d(p,c))^2}\\ 
    &=64\cdot \sqrt{(\slack(pq)+d(p,c))^2 -(d(p,c))^2 }\\
    &=64\cdot \sqrt{(d(p,q))^2 -(d(p,c))^2 }\\
    &=64\cdot \sqrt{\big((d(p,c))^2+(d(c,q))^2\big) -(d(p,c))^2 } \\
    &=64\cdot d(c,q) 
    =32\cdot \wts(L\cap \mathcal{E}_{pa,\eps/2})  .\qedhere
\end{align*}
\end{proof}

\section{Steiner Shallow Trees for Points in a Tile}
\label{sec:sector}

In this subsection, we prove \Cref{thm:Steiner}. 
By \Cref{thm:nonSteiner-reduction} and \Cref{lem:net},
we may assume that $P$ is a set of points in a trapezoid $\tau\in \mathcal{T}$ (defined in \Cref{ssec:SteinerReduction}), and $P$ is a centered $\eps$-net w.r.t.\ center $s\in P$. We may further assume w.l.o.g.\ that $\tau\subset [0,1]\times [-\sqrt{\eps},\sqrt{\eps}]$ and $s=(2,0)$; see \Cref{fig:sector}.
We first present an algorithm that constructs a Steiner tree $T$ 
for $P\cup \{s\}$, rooted at $s$. Then we show that $T$ is a $(1+O(\eps\log \eps^{-1}))$-ST 
and its weight is $O(\opt_\eps\cdot \log \eps^{-1})$, where $\opt_\eps$ denotes the minimum weight of a Steiner $(1+\eps)$-ST for $P$ rooted at $s$; and let $\OPT_\eps$ be one such Steiner $(1+\eps)$-ST.

\begin{figure}[htbp]
    \centering
    \includegraphics[width=.5\textwidth]{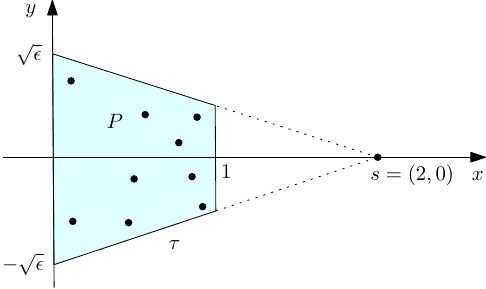}
    \caption{A point set $P$ in a trapezoid $\tau$ and a source $s=(2,0)$.}
    \label{fig:sector}
\end{figure}

\paragraph{Steiner shallow tree construction: Overview.}
We start with a brief overview of the ST construction and then present the algorithm in detail. 
We construct a Steiner graph $G$ on $P\cup \{s\}$ (which is not necessarily a tree), and then let $T$ be the shortest path tree of $G$ rooted at $s$ (that is, the single-source shortest path tree of $G$ with $s$ as the source). We can analyze the root-stretch in the graph $G$, and the weight of $T$ is upper-bounded by the weight of $G$.

For each point $p\in P$, we know that $\OPT_\eps$ contains a $ps$-path of length at most $(1+\eps)\cdot d(p,s)$ in the ellipse $\mathcal{E}_{ps}$ with major axis $(1+\eps)\cdot d(p,s)$. We associate each point $p\in P$ with vertical lines $L_0(p),\ldots , L_{k-1}(p)$, where $k=O(\log \eps^{-1})$. For $i=0,1,\ldots, k-1$, we choose a Steiner point $s_i$ in the line segments $\mathcal{E}_{ps}\cap L_i(p)$, and add the path $\pi_{ps}=(p,s_0,\ldots, s_{k-1},s)$ to the graph $G$.
We show (\Cref{lem:rootstretch}) that if the distances between the vertical lines $L_1,\ldots ,L_k$ increase exponentially, then $\wts(\pi_{ps})  = (1+O(\eps\log \eps^{-1}))\cdot d(p,s)$. 

Importantly, each line $L_i$ is associated with multiple points in $P$, and each Steiner point $s_i\in L_i$ also serves multiple points in $P$. For each line $L_i$, we choose the Steiner points in $L_i$ as a minimum hitting set for the line segments $\mathcal{E}_{ps}\cap L_i$ over all points $p\in P$ associated with line $L_i$. 
Since the intersection $\OPT_\eps\cap L_i$ is a hitting set for these intervals, we can charge the weight of $G$ to the weight of $\OPT_\eps$. The approximation factor $O(\log \eps^{-1})$, in the stretch and the weight, is the result of using $O(\log\eps^{-1})$ Steiner points, one in each line $L_i$, for $i=1,\ldots , k$.

\paragraph{Steiner shallow tree construction: Details.}
Assume w.l.o.g.\ that $\eps=4^{-(k+1)}$ for some $k\in \mathbb{N}$.
By \Cref{lem:approxellipse} each ellipse $\mathcal{E}_{ps}$ is contained in an ellipse $\mathcal{E}_{pb(p),2\eps}$, where the segment $pb(p)$ is horizontal and $d(p,b(p))\leq 2\cdot d(p,s)$.  As a shorthand notation, we use $\mathcal{E}_p:=\mathcal{E}_{p b(p),2\eps}$.

We define families of vertical lines. For every integer $i\geq 0$, let 
\begin{equation}\label{eq:lines}
    \mathcal{L}_i=\{ x= j  4^i\cdot \eps: j\in \mathbb{Z}\},
\end{equation}
that is, the distance between two consecutive vertical lines in $\mathcal{L}_i$ is $4^i\cdot \eps$. 
For every point $p\in P$, we recursively choose $L_i(p)\in \mathcal{L}_i$ for $i=0,1,\ldots , k-1$, as follows. 
Let $L_0(p)$ be the second line in $\mathcal{L}_0$ to the right of $p$. For $i=1,\ldots , k-1$, if $L_{i-1}(p)$ has already been chosen, let $L_i(p)$ be the second line in $\mathcal{L}_i$ to the right of $L_{i-1}(p)$. (We choose the second line in $\mathcal{L}_i$, rather than the first one, to ensure that the gaps between $L_{i-1}(p)$ and $L_i(p)$ grow exponentially in $i$; cf.~\Cref{obs:spacing}.)

Now consider a line $L\in \mathcal{L}_i$, and let $P(L)$
be the set of all points in $P$ associated with $L$. Formally, we put $P(L)=\{p\in P: L_i(p)=L\}$. Consider the set of intervals 
\[
    I(L)=\{L\cap \mathcal{E}_p: p\in P(L)\}.
\]
Let $H(L)$ be a \EMPH{minimum hitting set} (a.k.a., piercing set, stabbing set, or transversal) for $I(L)$: It is a minimum subset of $L$ that contains at least one  point in each interval in $I(L)$; it can be computed in $O(|I(L)|\log |H(L)|)$ time~\cite{DanzerG82,HochbaumM85,Nielsen00}.

Finally, we construct the Steiner graph $G$ as follows.
The vertex set of $G$ comprises $P$, $s$, and the points $H(L)$ for all $L\in \mathcal{L}_i$, $i=0,\ldots , k-1$.
The edges are defined as follows. 
For each point $p\in P$, consider the lines $L_0(p),\ldots , L_{k-1}(p)$ associated with $p$. For $i=0,1,\ldots ,k-1$ let $s_i(p)$ be an arbitrary point in $H(L_i(p)) \cap \mathcal{E}_p$. Add the edges of the path $\pi_{ps}=(p,s_0(p),\ldots, s_{k-1}(p),s)$ to $G$. 

As noted above, we let $T$ be the shortest path tree of $G$ rooted at $s$. 
This completes the construction of $T$.

\begin{observation}\label{obs:spacing}
 For $i=1,\ldots , k-1$, the distance between $L_{i-1}(p)$ and $L_i(p)$ is in the range $[4^i\, \eps, 2\cdot 4^i\, \eps]$.
\end{observation}
\begin{corollary}\label{cor:spacing}
     For $i=1,\ldots , k-1$, the distance between $p$ and $L_i(p)$ is in the range $[4^i\, \eps, \frac{32}{3}\cdot 4^i\, \eps]$.
\end{corollary}
\begin{proof}
By construction, the distance between $p$ and $L_0(p)$ is in the range $[\eps,2\eps]$. Summation over the lower and upper bounds in \Cref{obs:spacing} yields 
\begin{align*}
{\rm dist}(p,L_i(p)) &
    \geq \sum_{j=0}^i 4^j\cdot \eps
    = \frac{4^{i+1}-1}{4-1}\cdot \eps 
    \geq \frac{4^{i+1}-4^{i}}{3}\cdot \eps
    \geq 4^i\cdot \eps,\\
{\rm dist}(p,L_i(p)) &
    \leq  \eps+\sum_{j=0}^i 2\cdot 4^i\, \eps 
    = 2\left(1+4\cdot \frac{4^{i+1}-1}{4-1}\right)\eps 
    < \frac{32\cdot 4^i}{3}\cdot \eps.
    \qedhere
\end{align*}
\end{proof}

\paragraph{Root stretch analysis.} It is enough to analyze the root stretch in the graph $G$. Recall that in the construction of $G$, we have already built a path $\pi_{ps}=(p,s_0(p),\ldots,  s_{k-1}(p),s)$. We show that $\wts(\pi_{ps})\leq (1+O(\eps\log \eps^{-1}))\cdot d(p,s)$ (\Cref{lem:rootstretch}). We first need to estimate the slopes of the edges of $\pi_{ps}$. We start with a lemma about the lengths of the line segments $L_i(p)\cap \mathcal{E}_p$ for $i=0,\ldots , k-1$.

\begin{lemma}\label{lem:crosssection}
    For every point $p\in P$ if a vertical line $L$ is to the right of $p$ at distance $x\in [0,\frac23]$ from $p$, then 
    $\wts(L\cap \mathcal{E}_p)=\Theta(\sqrt{\eps\, x})$. 
    In particular, for every $i=0,\ldots , k-1$, we have $\wts(L_i(p)\cap \mathcal{E}_p)= \Theta(2^i \cdot \eps)$. 
\end{lemma}
\begin{proof}
Consider the unit disk $D=\{(x,y)\in \mathbb{R}^2: x^2+y^2=1\}$ and the line $L: x=-1+\varrho$ for $0\leq \varrho\leq 1$; see \Cref{fig:section}. Then for $\varrho\in [0,1]$, we obtain 
\[
    \wts(D\cap L)
    =2\cdot \sqrt{1^2-(1-\varrho)^2}
    =  2\cdot \sqrt{2\varrho-\varrho^2}
    = 2^{3/2}\cdot \sqrt{\varrho - \frac{\varrho^2}{2}} 
    =\Theta(\sqrt{\varrho}) .
\]
The ellipse $\mathcal{E}_p=\mathcal{E}_{pb,2\eps}$ is an affine image of the unit disk $D$: 
The affine transformation $f(x,y)=(\frac{M}{2}x,\frac{m}{2}y)$ takes $D$ to $\mathcal{E}_p$, where $M$ and $m$ are the major and minor axes of $\mathcal{E}_p$.
We estimate $M$ and $m$ up to constant factors. Since $P\subset [0,1]\times [-\sqrt{\eps},\sqrt{\eps}]$ and $s=(2,0)$, then we have $1\leq |x(p)-x(s)|\leq 2$. 
By \Cref{lem:approxellipse}, we have $1\leq d(p,b(p))\leq 4$. The major axis of $\mathcal{E}_p$ is $M=(1+2\eps)\cdot d(p,b(p))$ so $1< M\leq 4+8\eps$; and its minor axis is $m=2\cdot \sqrt{(1+2\eps)^2-1}\cdot d(p,b(p))=4\cdot \sqrt{\eps+\eps^2}\cdot d(p,b(p))$, and so $4\cdot \sqrt{\eps}\leq m\leq 16\cdot \sqrt{\eps + \eps^2} \le 16\sqrt{2}\cdot \sqrt{\eps}$. 

\begin{figure}[!htb]
    \centering
    \includegraphics[width=.6\textwidth]{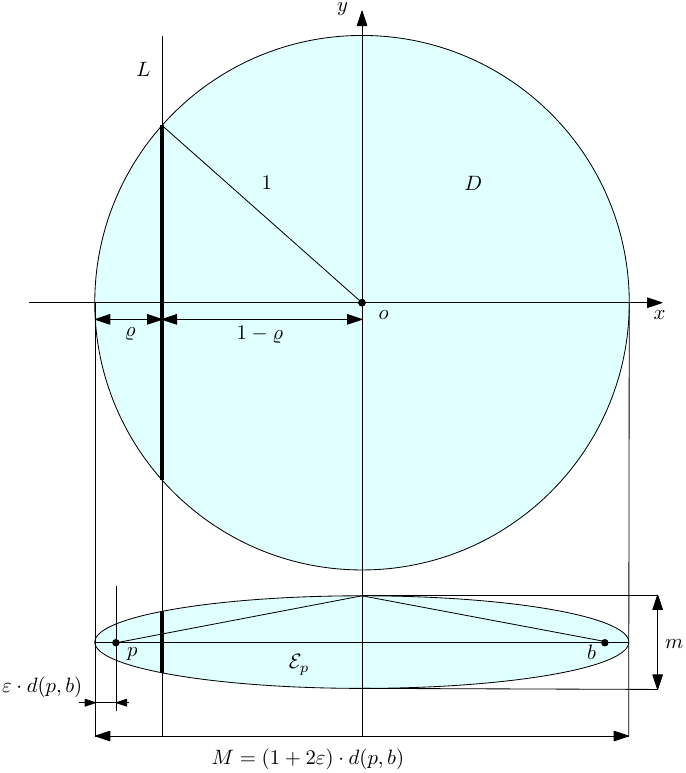}
    \caption{A unit disk $D$, a cross section $L\cap D$, and an affine transformation to the ellipse $\mathcal{E}_p$.}
    \label{fig:section}
\end{figure}

The focus $p$ is at distance $\eps\cdot d(p,b(p))$ from the leftmost point of $\mathcal{E}_p$,
where $\eps \le \eps\cdot d(p,b(p))\leq 4\eps$. By \Cref{cor:spacing},
the distance between $L_i(p)$ and the leftmost point of $\mathcal{E}_p$ is at least $\eps+4^i\eps=\Omega(4^i\eps\cdot M)$ and 
at most $4\eps+\frac{32}{3}\cdot 4^i\eps=O(4^i\eps\cdot M)$. Overall, this distance is $\Theta(4^i\eps\cdot M)$. 

The inverse transformation $f^{-1}$ takes $\mathcal{E}_p$ to the unit disk $D$, and the line $L_i(p)$ to a vertical line $L:x=-1+\varrho$, where $\varrho=\Theta(4^i\eps)$. As noted above, 
we have $\wts(L\cap D) =\Theta(\sqrt{\varrho})=\Theta(2^i\cdot \sqrt{\eps})$, consequently $\wts(L_i(p)\cap \mathcal{E}_p) =
\Theta(m\cdot \sqrt{\varrho})=\Theta(\sqrt{\eps}\cdot 2^i\cdot \sqrt{\eps}) = \Theta(2^i \cdot \eps)$.
\end{proof}
\begin{corollary}\label{cor:crosssection}
For every $p\in P$, consider the path $\pi_{ps}=(p,s_0,\ldots,  s_{k-1},s)$. For all $i=1,\ldots , k-1$, we have 
$|\slope(s_{i-1}s_i)|\leq O(2^{-i})$.
\end{corollary}
\begin{proof}
    \Cref{obs:spacing} gives a lower bound $|x(s_{i-1})-x(s_i)|\geq 4^i\cdot \eps$, and \Cref{lem:crosssection} gives an upper bound $|y(s_{i-1})-y(s_i)|\leq O(2^i\cdot \eps)$. The combination of these bounds yields
    \[
    |\slope(s_{i-1}s_i)| =  
    \frac{|y(s_{i-1})-y(s_i)|}{|x(s_{i-1})-x(s_i)|}
    \le \frac{O(2^i\cdot \eps)}{4^i\cdot \eps} 
    \leq O(2^{-i}). \qedhere
    \]
    
\end{proof}

\begin{lemma}\label{lem:rootstretch}
    For every $p\in P$, we have 
    $\wts(\pi_{ps})\leq (1+O(\eps\log \eps^{-1}))\cdot d(p,s)$.
\end{lemma}
\begin{proof}
By construction, we have $\pi_{ps}=(p,s_0,s_1,\ldots s_{k-1},s)$. We partition $\pi_{ps}$ into three parts: 
$ps_0$, $(s_0,s_1,\ldots s_{k-1})$ and $s_{k-1}s$, and bound the weight of each part separately. 

First we estimate the weight of the first edge $ps_0$. 
Recall that $p\in [0,1]\times [-\sqrt{\eps},\sqrt{\eps}]$ and $s=(2,0)$. By construction, we have $|x(p)-x(s_0)|\leq \eps$, and \Cref{lem:crosssection} gives $|y(p)-y(s_0)|\leq O(\eps)$. Therefore, we have $\wts(ps_0)=O(\eps)\leq O(\eps)\cdot d(p,s)$.

We can now bound the weight of the subpath $(s_0,s_1,\ldots , s_{k-1})$ of $\pi_{ps}$. We use \Cref{lem:slack}, \Cref{obs:spacing}, and \Cref{cor:crosssection} for each edge:
\begin{align*}
\wts((s_0,s_1,\ldots , s_{k-1})) 
    &=\sum_{j=1}^k \wts(s_{j-1} s_j)
     \leq \sum_{j=1}^{k-1} 
     \left(1+\frac{(\slope(s_{j-1}s_j))^2}{2}\right) \cdot |\proj(s_{j-1}s_j)| \\
    & = \sum_{j=1}^{k-1} |\proj(s_{j-1}s_j)|  + \sum_{j=1}^{k-1} O(4^{-j}) \cdot |\proj(s_{j-1}s_j)| \\
    & \leq |\proj(s_0 s_{k-1})| + \sum_{j=1}^{k-1} O(4^{-j} ) \cdot 2\cdot 4^j\cdot \eps
    =|\proj(s_0 s_{k-1})| + O(k\eps) \\ 
    &=|\proj(s_0 s_{k-1})| + O\left(\eps\log \eps^{-1}\right) 
    =|\proj(s_0 s_{k-1})|  \left(1+ O\left(\eps\log \eps^{-1}\right) \right) .
\end{align*}

Finally, we estimate the weight of the last edge, $s_{k-1}s$, of $\pi_{ps}$. Combining \Cref{cor:spacing} with $\eps=4^{-(k+1)}$, we have 
\begin{equation}\label{eq:dist}
    {\rm dist}(p,L_{k-1}(p)) < \frac{32}{3}\cdot 4^{k-1}\cdot \eps = \frac{2}{3\eps}\cdot \eps = \frac{2}{3} .
\end{equation}
Since $x(p)\in [0,1]$ and $s=(2,0)$, then $|x(s_{k-1})-x(s)|\geq \frac13$. \Cref{lem:crosssection} gives $|y(s_{k-1})-y(s)|=|y(s_{k-1})|\leq O(2^{k-1} \cdot \eps) = O\left(\frac{1}{\sqrt{\eps}} \cdot \eps\right) = O(\sqrt{\eps})$. Consequently, $\slope(s_{k-1}s)=O(\sqrt{\eps})$, and \Cref{lem:slack} yields 
$\wts(s_{k-1}s)\leq (1+O(\eps))\cdot |\proj(s_{k-1}s)|$.

The sum of the weights of the three parts is 
\begin{align*}
\wts(\pi_{ps}) & =
    \wts(ps)+ \wts((s_0,s_1,\ldots , s_{k-1}))  +\wts(s_{k-1}s) \\
    &\leq O(\eps)\cdot d(p,s) +
   \left(1+ O\left(\eps\log \eps^{-1}\right) \right) |\proj(s_0 s_{k-1})|   
    + (1+O(\eps))\cdot |\proj(s_{k-1}s)|\\
    &\leq O(\eps)\cdot d(p,s) +
   \left(1+ O\left(\eps\log \eps^{-1}\right) \right) |\proj(s_0 s)| \\
    &\leq    \left(1+ O\left(\eps\log \eps^{-1}\right) \right) d(p,s) . \qedhere
\end{align*}
\end{proof}

\paragraph{Weight analysis.}
Recall that we constructed a Steiner graph $G=\bigcup _{p\in P} \pi_{ps}$, and the final ST is a shortest-path tree rooted at $s$. We give an upper bound for $\wts(G)$. Similarly to the root stretch analysis, we decompose $\pi_{ps}$ into three parts, $ps_0$, $(s_0,s_1,\ldots ,s_{k-1})$, and $s_{k-1}s$, and then bound the weight of the union of each.

We first analyze the total weight of the union of edges $s_{i-1}(p)\, s_i(p)$ of the paths $\pi_{ps}$ over all $p\in P$. For lines $L_a\in \mathcal{L}_{i-1}$ and $L_b\in \mathcal{L}_i$, let $G(L_a,L_b)$ denote the set of all edges $uv\in E(G)$ such that $u\in I(L_a)$ and $v\in I(L_b)$. Recall that $\OPT_\eps$ is a minimum-weight Steiner $(1+\eps)$-ST for $P\cup \{s\}$ rooted at $s$ (i.e., $\wts(\OPT_\eps)=\opt_\eps$). We use the notation $\OPT_\eps\cap \llbracket L_a,L_b\rrbracket$ for the part of $\OPT_\eps$ clipped in the vertical strip $\llbracket L_a,L_b\rrbracket$. 

Our main lemma for the weight analysis is as follows.
\begin{lemma}\label{lem:weight}
For every $p\in P$ and $i\in \{1,\ldots , k-1\}$,
 we have 
 \[
    \wts(G(L_{i-1}(p),L_i(p)))\leq O(\wts(\OPT_\eps\cap \llbracket L_{i-1}(p),L_i(p)\rrbracket)).
\]
\end{lemma}

Before the proof of \Cref{lem:weight}, we show that it implies the bound on the total weight of the interior edges of the paths $\pi_{ps}$ over all $p\in P$.

\begin{corollary} \label{cor:interior}
We have
$\wts\left(\bigcup_{p\in P} (s_0(p),\ldots,  s_{k-1}(p)) \right)   
    \leq  O(\log \eps^{-1}) \cdot \wts(\OPT) = O(\log \eps^{-1}\cdot \opt_\eps)$.
\end{corollary}
\begin{proof}
For $i=1,\ldots, k-1$, let 
\[
    \mathcal{X}_i=\{ (L_{i-1}(p),L_i(p)) : p\in P\}
\]
and let $\mathcal{X}=\bigcup_{i=1}^{k-1}\mathcal{X}_i$. 
Note that for every $i\in \{1,\ldots ,k-1\}$ and 
every line $L_b\in \mathcal{L}_i$, there are only $O(1)$ lines $L_a\in \mathcal{L}_{i-1}$ such that $(L_a,L_b)\in \mathcal{X}_i$. 
Indeed, if $L_a=L_{i-1}(p)$ and $L_b=L_i(p)$ , then the distance between $L_a$ and $L_b$ is at most $2\cdot 4^i\cdot \eps$ by \Cref{obs:spacing}; and the distance between two consecutive lines in $\mathcal{L}_{i-1}$ is $4^{i-1}\cdot \eps$. 

Since $L_b$ is the right boundary of the strip $\llbracket L_a,L_b\rrbracket$, it follows that any point in the plane is contained in $O(1)$ strips $\llbracket L_a,L_b\rrbracket$ where $(L_a,L_b)\in \mathcal{X}_i$; and any point in the plane is contained in $O(k)=O(\log\eps^{-1})$
strips $\llbracket L_a,L_b\rrbracket$, where  $(L_a,L_b)\in \mathcal{X}$. 

Now we can apply \Cref{lem:weight}.
\begin{align*}
\wts\left(\bigcup_{p\in P} (s_0(p),\ldots,  s_{k-1}(p)) \right)  
    & = \wts\left( \bigcup_{i=1}^{k-1} \bigcup_{p\in P} s_{i-1}(p)s_i(p)\right)
    \leq \sum_{i=1}^{k-1}  \wts\left(  \bigcup_{p\in P} s_{i-1}(p)s_i(p)\right)\\
    &= \sum_{i=1}^{k-1} \sum_{(L_a,L_b)\in \mathcal{X}_i}\wts\left( G(L_a,L_b)\right)\\
    &\leq \sum_{i=1}^{k-1} \sum_{(L_a,L_b)\in \mathcal{X}_i} O(\wts(\OPT_\eps\cap \llbracket L_a,L_b\rrbracket))\\
    &\leq O(\log \eps^{-1}) \cdot \wts(\OPT_\eps). 
    \qedhere
\end{align*}
\end{proof}

It remains to prove \Cref{lem:weight}. We do this in a sequence of lemmas, using geometric properties of ellipses and interval graphs. We start with an easy observation.
\begin{observation}\label{obs:bracket}
    For every $p\in P$ 
    and every $i\in \{1,\ldots , k-1\}$, we have 
\[
\wts(\OPT_\eps\cap  \llbracket L_{i-1}(p),L_i(p)\rrbracket\cap \mathcal{E}_p) \geq {\rm dist}(L_{i-1}(p),L_i(p)). 
\]
\end{observation}
\begin{proof}
  Since $\OPT_\eps$ is a $(1+\eps)$-ST for $P$ rooted at $s$, it contains a $ps$-path of length at most $(1+\eps)\cdot d(p,s)$. Any such path lies in the ellipse $\mathcal{E}_{ps,\eps}\subset \mathcal{E}_{pb(p),2\eps}=\mathcal{E}_p$.
 Every $ps$-path crosses both $L_{i-1}(p)$ and $L_i(p)$. 
 Its subpath between the two closest intersection points with these two lines is contained in the strip $\llbracket L_{i-1}(p),L_i(p)\rrbracket$, and
 the weight of this subpath is at least ${\rm dist}(L_{i-1}(p),L_i(p))$.
\end{proof}

Recall that on every line $L\in \mathcal{L}$, 
$I(L)=\{L\cap \mathcal{E}_p: p\in P(L)\}$, and $H(L)$ is a minimum hitting set for the intervals $I(L)$. 

\begin{lemma}\label{lem:boundingbox}
For every $i\in \{1,\ldots , k-1\}$ and every $p\in P$, there exists a set $Q\subset P(L_{i-1}(p))$ of size $|Q|\geq \Omega(|H(L_{i}(p))|)$ 
such that the regions in $\{\mathcal{E}_q\cap \llbracket L_{i-1}(p),L_i(p)\rrbracket: q\in Q\}$ are disjoint. 
\end{lemma}
\begin{proof}
Each interval in $I(L_{i}(p))$ has weight $\Theta(2^{i}\cdot \eps)$ by \Cref{lem:crosssection}. Let $c\geq 1$ be the ratio of the maximum to the minimum weight of an interval in $I(L_{i}(p))$ 

Each interval in $I(L_{i}(p))$ is of the form $L_{i}(p)\cap\mathcal{E}_r$ for some point $r\in P(L_{i}(p))$. For every $r\in P(L_{i}(p))$, let $B_r$ be the axis-aligned bounding box of $\mathcal{E}_r\cap \llbracket L_{i-1}(p),L_i(p)\rrbracket$; see \Cref{fig:packing}. Note that the major axis of $\mathcal{E}_r$ is horizontal, and its minor axis is to the right of $L_i$ by \Cref{eq:dist}. Consequently, for every vertical line $L\subset \llbracket L_{i-1}(p),L_i(p)\rrbracket$, we have $\wts(L\cap \mathcal{E}_r)\leq\wts(L_i(p)\cap \mathcal{E}_r)$, which holds with equality for $L=L_i(p)$. In particular, the height of $B_r$ is $\wts(L_i(p)\cap \mathcal{E}_r)=\Theta(2^i\cdot \eps)$. 

\begin{figure}[htbp]
    \centering
    \includegraphics[width=.85\textwidth]{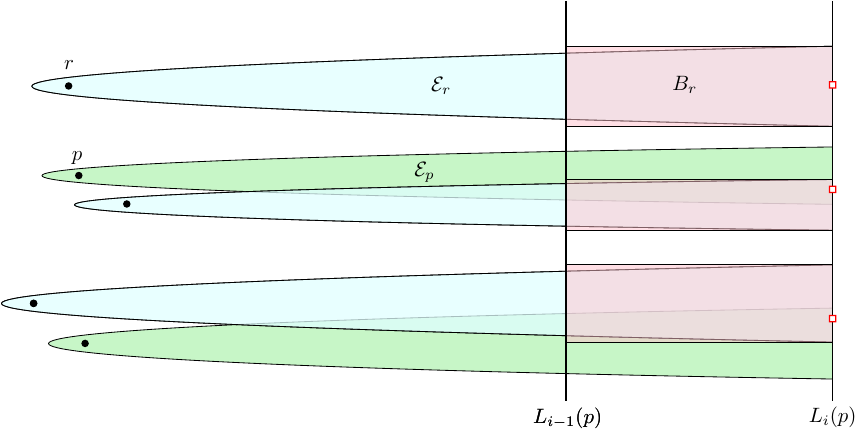}
    \caption{A point $p\in P$ and ellipses $\mathcal{E}_r$ for five points in $P(L_{i-1}(p))$. The boxes $B_r$ corresponding to the three light blue ellipses are disjoint. A minimum hitting set $H(L_{i}(p))$ has size 3.}
    \label{fig:packing}
\end{figure}

Now consider a maximum set $Q_0\subset P(L_{i}(p))$
such that the intervals $\{L_{i}(p)\cap \mathcal{E}_q: q\in Q_0\}$ are disjoint.
Note that $|Q_0|=|H(L_{i}(p))|$, that is, the maximum independent set has the same size as a hitting set for intervals in a line. 

A set of disjoint intervals along the line $L_{i}(p)$ have a well-defined total order, which defines a total order on $Q_0$. Let $Q\subset Q_0$ be the subset that corresponds to the first and every $\lceil c\rceil$-th interval.  Then $|Q_0| = \Theta(|Q|)=\Theta(|H(L_{i}(p))|)$.
Furthermore, the intervals 
$\{L_i(p)\cap \mathcal{E}_q: q\in Q\}$ are disjoint. This, in turn, implies that the boxes $\{B_q: q\in Q\}$ are disjoint. By the definition of $B_q$, this implies that the regions $\{\mathcal{E}_q\cap \llbracket L_{i-1}(p),L_i(p)\rrbracket : q\in Q\}$ are also disjoint. 
\end{proof}

\begin{corollary}\label{cor:boundingbox}
For every $p\in P$ and $i\in \{1,\ldots , k-1\}$,
 we have 
 \[
    \Omega(|H(L_i(p))|)\cdot {\rm dist}(L_{i-1}(p),L_i(p)) 
    \leq \wts \big(\OPT_\eps\cap \llbracket L_{i-1}(p),L_i(p)\rrbracket \big) .
\]
\end{corollary}
\begin{proof}
    This follows immediately from the combination of \Cref{obs:bracket} and \Cref{lem:boundingbox}. Specifically, \Cref{lem:boundingbox}
    yields a set $Q\subset L_i(p)\subset P$ of point for which the regions 
    $\{\mathcal{E}_q\cap \llbracket L_{i-1}(p),L_i(p)\rrbracket: q\in Q\}$ 
    are disjoint. Each region $\{\mathcal{E}_q\cap \llbracket L_{i-1}(p),L_i(p)\rrbracket$ contains a path in $\OPT_\eps$ of weight 
    at least ${\rm dist}(L_{i-1}(p),L_i(p))$ by \Cref{obs:bracket}.
    Consequently, we obtain 
\begin{align*}
    \wts(\OPT_\eps \cap \llbracket L_{i-1}(p),L_i(p)\rrbracket)
    &\geq \sum_{q\in Q} \wts\left(\big(\OPT_\eps\cap \mathcal{E}_r\big)
    \cap \llbracket L_{i-1}(p),L_i(p)\rrbracket \right)\\
    &\geq |Q| \cdot {\rm dist}(L_{i-1}(p),L_i(p)) \\
    &\geq \Omega(|H(L_i(p)|)\cdot {\rm dist}(L_{i-1}(p),L_i(p)) .
     \qedhere
\end{align*}
\end{proof}

\begin{lemma}\label{lem:graph}
For every $p\in P$ and $i\in \{1,\ldots , k-1\}$,
 we have 
 \[
    \wts(G(L_{i-1}(p),L_i(p)))\leq O(|H(L_{i}(p)|)\cdot {\rm dist}(L_{i-1}(p),L_i(p)) .
\]
\end{lemma}
\begin{proof}
We first show that every edge in $G(L_{i-1}(p),L_i(p))$
has weight $\Theta({\rm dist}(L_{i-1}(p),L_i(p)))$.
Every edge of $G(L_{i-1}(p),L_i(p))$ is the edge $s_{i-1}(r)s_i(r)$ of $\pi_{rs}$ for some point $r\in P(L_{i-1}(p))$. By \Cref{cor:crosssection},
$|\slope(s_{i-1}(r)s_i(r))|\leq O(2^{-i})$, where $i=O(\log_4\eps^{-1})$. Consequently, we have
$|\slope(s_{i-1}(r)s_i(r))|\leq O(1)$ for all $i=1,\ldots, k-1$. Therefore, 
\[
    \wts(s_{i-1}(r) s_i(r))
    \leq O(1)\cdot |\proj(s_{i-1}(r) s_i(r))| 
    \leq O( {\rm dist}(L_{i-1}(p),L_i(p))).
\]

Next, we show that every vertex $v\in H(L_i(p))$ of $G(L_{i-1}(p),L_i(p))$ has $O(1)$ degree. \Cref{lem:crosssection} gives $\wts(L_{i-1}(p)\cap \mathcal{E}_r)= \Theta(2^i)$ and $\wts(L_{i}(p)\cap \mathcal{E}_r)= \Theta(2^{i-1})$ for all $r\in P(L_{i-1}(p))$. Note also that both $L_{i-1}(p)\cap \mathcal{E}_r$ and $L_{i}(p)\cap \mathcal{E}_r$ have a reflection symmetry in the horizontal major axis of $\mathcal{E}_r$.
This implies that $\bigcup \{L_{i}(p)\cap \mathcal{E}_r: r\in P(L_{i}(p)) \mbox{ \rm and } v \in \mathcal{E}_r\}$ is contained in an interval of length $2\cdot \Theta(2^{i}) =\Theta(2^i)$ centered at $v$. Similarly, $\bigcup \{L_{i-1}(p)\cap \mathcal{E}_r: r\in P(L_{i-1}(p)) \mbox{ \rm and } v \in \mathcal{E}_r\}$ is contained in an interval $I\subset L_{i-1}$ of length $2\cdot \Theta(2^{i-1}) =\Theta(2^i)$. Since $\wts(L_{i-1}(p)\cap \mathcal{E}_r)= \Theta(2^{i})$ for all $r\in P(L_{i-1}(p))$, then the hitting set $H(L_{i-1}(p))$ contains $O(1)$ points in the interval $I$. In the graph $G(L_{i-1}(p),L_i(p))$, all neighbors of vertex $v$ are in the interval $I$ and in $H(L_{i-1}(p))$. This proves that $v$ has $O(1)$ neighbors in $G(L_{i-1}(p),L_i(p))$.

Overall, every edge of $G(L_{i-1}(p),L_i(p))$ has weight $\Theta( {\rm dist}(L_{i-1}(p),L_i(p)))$; and every vertex $v\in H(L_{i}(p))$ of $G(L_{i-1}(p),L_i(p))$ has $O(1)$ degree. Since $G(L_{i-1}(p),L_i(p))$ is a bipartite graph with partite sets $H(L_{i-1}(p))$ and $H(L_{i}(p))$, then the total weight of $G(L_{i-1}(p),L_i(p))$ is $O(|H(L_{i}(p)|)\cdot {\rm dist}(L_{i-1}(p),L_i(p))$, as claimed.
\end{proof}

We are now ready to prove the key lemma of the weight analysis. 

\begin{proof}[Proof of \Cref{lem:weight}.]
We combine \Cref{lem:graph} and \Cref{cor:boundingbox}. 
For every $p\in P$ and $i\in \{1,\ldots , k-1\}$,
 we obtain
\begin{align*}
\wts(G(L_{i-1}(p),L_i(p)))
    &\leq O(|H(L_{i-1}(p)|)\cdot {\rm dist}(L_{i-1}(p),L_i(p)) \\
    &\leq O(\wts(\OPT\cap \llbracket L_{i-1}(p),L_i(p)\rrbracket)). \qedhere
\end{align*} 
\end{proof}

Finally, we summarize the weight analysis in the following lemma.

\begin{lemma}
    For an $\eps$-cnet $P\subset [0,1]\times [-\sqrt{\eps},\sqrt{\eps}]$ and $s=(0,2)$,
    the weight of the Steiner graph $G$ is 
    $O(\log \eps^{-1})\cdot \opt_{\eps}$, where 
    $\opt_{\eps}$ denotes the minimum weight of a Steiner $(1+\eps)$-ST for $P$.
\end{lemma}
\begin{proof}
We have shown that for every $p\in P$, the first edge of $\pi_{ps}$ has weight $O(\eps)$. Consequently, the total weight of the union of first segments of $\pi_{ps}$ over all $p\in P$ is $O(\eps\cdot |P|)$. Since $P$ is an $\eps$-cnet, then this is bounded by $O(\opt_\eps)$. 

By \Cref{cor:interior}, the
 total weight of the interior edges of the paths $\pi_{ps}$ over all $p\in P$, namely, $\wts\left(\bigcup_{p\in P} (s_0(p),\ldots,  s_{k-1}(p)) \right)$ 
is bounded by $O(\log \eps^{-1}) \cdot \opt_\eps$.

Finally, we bound the total weight of the last edges of the paths $\pi_{ps}=(p,s_0(p),\ldots , s_{k-1}(p),s)$ for all $p\in P$. By \Cref{obs:spacing},
there are $O(1)$ vertical lines in $\mathcal{L}_{k-1}$ between $o=(0,0)$ and $s=(2,0)$. By \Cref{cor:crosssection}, we have $\wts(\mathcal{E}_p\cap L_{k-1}(p))=\Theta(\sqrt{\eps})$. Consequently, each piercing set $H(L_{k-1}(p))$ has $O(1)$ size. 
Overall, there are $O(1)$ distinct segments in $\{s_{k-1}(p)s: p\in P\}$, and the weight of each segment is $O(1)$. Consequently, the total weight of these segments is also $O(1)$. Since ${\rm dist}(P,s)\geq 1$, then $\wts(OPT_\eps)\geq 1$, and so $\wts( \{s_{k-1}(p)s: p\in P\} ) \leq O(1)\leq O(\wts(\OPT_\eps))$. 
\end{proof}

\paragraph{Running time.} 
Let $P$ be a set of $n$ points in the plane, including a source $s\in P$, and let $\eps>0$.
The $n$ points in $P$ can be assigned to tiles in $O(n)$ time (based on the direction and length of the line segment $ps$). By \Cref{thm:nonSteiner-reduction} and \Cref{lem:net}, we can reduce the problem to $\eps$-cnets in a tile in $O(n\log n)$ time. 
If $P_i$ is a $\eps$-cnet in a tile $\tau_i$, then $n_i=|P_i|=O(\eps^{-2})$ by a straightforward volume argument; but we also have $\sum_{i} n_i \leq \sum_i |P\cap \tau_i|\leq n$. We show that our single-tile algorithm runs in $O(n_i \log n_i\cdot {\rm polylog}(\eps^{-1}))$ time for $P_i$. Summation over all tiles gives $O(n \log n\cdot {\rm polylog}(\eps^{-1}))$ time. By  \Cref{thm:nonSteiner-reduction}, the overall running time of for the entire point set is $O(n \log n\cdot{\rm polylog}(\eps^{-1}))$. 

Consider an $\eps$-cnet $P_i$ of size $n_i$ in a tile $\tau_i$. Our algorithm associates each point $p\in P_i$ to $O(\log \eps^{-1})$ vertical lines $L_0(p),\ldots , L_{k-1}(p)$. These lines can be computed recursively from the $x$-coordinate of $p$, in $O(1)$ time per line, hence in $O(n_i\log \eps^{-1})$ time overall. Each point $p\in P_i$ contributes an interval $\mathcal{E}_p\cap L$ to an instance of the hitting set problem along each associated line $L$. The total number of intervals over all instances is $O(n_i\log \eps^{-1})$, but each instance involves $O(n_i)$ intervals. The hitting set problem on $m$ intervals in a line can be solved exactly in $O(m\log m)$ time; so we can solve all hitting set instances in $O(n_i\log n_i\log \eps^{-1})$ total time. The total size of all hitting sets is $O(n_i\log \eps^{-1})$.
We construct the Steiner graph $G$ by connecting Steiner points in consecutive levels. We have shown (\Cref{lem:graph}) that every Steiner point has bounded degree in $G$, so $G$ can be computed in $O(n_i\log \eps^{-1})$ time. Finally, we return a shortest-path tree $H$ rooted at $s$ in $G$, which can be computed in $O((n_i \log \eps^{-1})\log (n_i \log \eps^{-1})) \leq O(n_i\log n_i\cdot {\rm polylog}(\eps^{-1}))$ time.

\section{Shallow Trees for Points in a Tile}
\label{sec:general-sector}

In this section, we prove \Cref{thm:main}. By \Cref{thm:nonSteiner-reduction} and \Cref{lem:net},
we may reduce to the problem of finding a tile-restricted $(1 + \eps)$-ST for points in a tile $\tau$. 
Let $P$ be a set of $n$ points in the plane, $\tau\in \mathcal{T}$ a tile  (defined in \Cref{ssec:SteinerReduction}), and $N_\tau$ an $\eps$-cnet for the point set $P\cap \tau$.
We first present an algorithm that constructs a net-restricted ST $T$ for $N_\tau\subset P$. Then we show that $T$ is a net-restricted $(1+O(\eps\log \eps^{-1}))$-ST 
and its weight is $O(\opt\cdot \log^2(\eps^{-1}))$, where $\opt$ denotes the minimum weight of a net-restricted $(1+\eps)$-ST for $N_\tau\subset P$; and let $\OPT$ be one such a $(1+\eps)$-ST.

\begin{lemma}
    Let $\opt(\tau)$ be the weight of the lightest net-restricted $(1 + \eps)$-ST for $N_\tau \cup \{s\}$. There is a polynomial time algorithm returns a net-restricted $(1 + \eps \cdot \log(\eps^{-1}))$ ST $T$ for $N_\tau \cup \{s\}$ of weight  $O(\log^2(\eps^{-1})\cdot \opt(\tau))$.   
\end{lemma}

\paragraph{Shallow tree construction algorithm.}
Similar to \Cref{sec:sector}, we may assume w.l.o.g.\ that $\tau\subset [0,1]\times [-\sqrt{\eps},\sqrt{\eps}]$ and $s=(2,0)$ (as in \Cref{fig:sector}).
For each point $p\in N_\tau$, we use the same set of vertical lines $L_0(p), L_1(p), \ldots$ defined in~\Cref{eq:lines}. For each point $p$ and each integer $i$, let $\mathcal{E}_{p}(i)$ denote the intersection of $\mathcal{E}_{ps}$ with the strip between the two lines $L_{i}(p)$ and $L_{i + 1}(p)$. Let $B_{i}(p)$ be the smallest axis-parallel rectangle that contains $\mathcal{E}_{p}(i)$. The point $p$ is associated with the rectangles $B_0(p), B_1(p), \ldots ,B_k(p)$ with $k = O(\log{\eps^{-1}})$. A rectangle $R$ \emph{lies on} two distinct vertical lines $L_1$ and $L_2$ if two of its vertical sides are contained in $L_1$ and $L_2$. 

We construct minimum hitting sets in $k$ iterations. Starting from $i = 0$, let $\mathcal{R}_i$ be the set of rectangles $\{B_i(p): p \in P_\tau\}$. Consider every pair of lines $L_{i} \in \mathcal{L}_i, L_{i + 1} \in \mathcal{L}_{i + 1}$ such that there are some rectangles in $\mathcal{R}_i$ lying on $L_i$ and $L_{i + 1}$. Let $\mathcal{R}_i(L_i, L_{i + 1})$ be the set of rectangles in $\mathcal{R}_i$ lying on $L_i$ and $L_{i + 1}$ that have nonempty intersection with $P$. We find the minimum hitting set of $\mathcal{R}_i(L_i, L_{i + 1})$ from $P$.
For each point $u$ in the hitting set, let $\iter_p(u) = i$ for every $p \in P(L_i)$ (and define $\iter_p(s) := k + 1$ for all $p\in P_\tau$).

We build a graph $G$ on $P$ as follows: For every point $p\in N_\tau$, let $s_0, s_1, \ldots s_{h}$ be a sequence of hitting points so that for $i=0,1,\ldots h$, if $B_i(p)\cap P\neq \emptyset$, then we choose an arbitrary point from the minimum hitting set of $\mathcal{R}(L_i(p),L_{i+1}(p))$ that hits $B_i(p)$ (otherwise no point is chosen from $B_i(p)$. We add the path $(p, s_0, s_1, \ldots ,s_{h}, s)$ to $G$. The path $(p, s_0, s_1, \ldots ,s_{h}, s)$ is called the \EMPH{candidate path} of $p$ in $G$.

For each point $p \in N_\tau$, let $H(p) = (p, s_0, s_1, \ldots s_{h}, s)$ be the candidate path of $p$  in $G$. We prune the path $H(p)$ to eliminate some of the vertices: Initialize $H'(p):=H(p)$. While there exist two consecutive interior vertices $v$ and $v'$ in $H'(p)$ such that $\iter_p(v') = \iter_p(v) + 1$, then let the first such pair be $v,v'$ and eliminate $v'$ from $H'(p)$. 
Assume that $H'(p) = (p, s'_0, s'_1, \ldots s'_{l}, s)$ is the result of the pruning process. Observe that $H'(p)$ is a path from $p$ to $s$. We call $H'(p)$ the \EMPH{approximate path} of $p$. Let $G'$ be the union of paths $H'(p)$ for all $p\in N_\tau$. We return the shortest-path tree $\mathsf{SPT}(G')$ of $G'$ rooted at $s$.

\paragraph{Stretch Analysis. }
Let $p$ be a point in $N_{\tau}$, and let $\pi = (p, s_0, s_1, \ldots, s_{h - 1}, s)$ be the approximate path of $p$ in $G'$. Then, we have $\iter_p(s_j) + 1 < \iter_p(s_{j+1})$ for every index $j$. We show that this implies $\slope(s_is_{i + 1}) \leq O(2^{-i})$. \begin{claim}
    For any index $j$, let $\iter_p(s_j) = i_1$ and $\iter_p(s_{j + 1})= i_2$ with $i_2 > i_1 + 1$. Then, $\slope(s_j s_{j + 1}) = O(2^{-j})$.
\end{claim}
\begin{proof}
    From \Cref{lem:crosssection}, $|x(s_{j + 1}) - x(s_{j})| = \Theta\left(\sum_{i_1 < i < i_2}4^i\eps\right) = \Theta(4^{i_2}\eps)$ and $|y(s_{j + 1}) - y(s_{j})|$ = $\Theta(2^{i_2}\eps - 2^{i_1}\eps) = \Theta(2^{i_2}\eps)$. Thus, $\slope(s_j s_{j + 1}) = O(2^{-i_2}) = O(2^{-j})$ since $i_2 > j$.
\end{proof}

Using the same proof as \Cref{lem:rootstretch}, we obtain the following lemma.
\begin{lemma}
    $\wts(\pi) \leq (1 + \eps \cdot \log{\eps^{-1}})\cdot d(u, s)$.
\end{lemma}
Since $\mathsf{SPT}(G')$ is the shortest-path tree of $G'$, it also contains a $ps$-path of weight at most $(1+O(\eps\log \eps^{-1}))\cdot d(p,s)$ for all $p\in N_\tau$.

\paragraph{Weight Analysis. } 
We first show that $\wts(G)\leq O(\opt\cdot \log\eps^{-1})$ (\Cref{lem:bound_weight}),
and then prove that $\wts(G')\leq O(\wts(G)\cdot \log\eps^{-1})$ (\Cref{lem:skipping}).
Combined with the trivial bound $\wts(\mathsf{SPT}(G'))\leq \wts(G')$, we 
obtain $\wts(\mathsf{SPT}(G'))\leq O(\opt \cdot \log^2\eps^{-1}))$, as required. 
We first show the following claims:
\begin{claim}\label{cl:perlevel}
    Let $\llbracket L_i, L_{i+1} \rrbracket$ be a strip for some $L_i \in \mathcal{L}_i$ and $L_{i + 1} \in \mathcal{L}_{i + 1}$. Assume that $X\subseteq N_{\tau}(L_i)$ and a point $v\in P\cap \llbracket L_i, L_{i+1} \rrbracket$ lies in the ellipse $\mathcal{E}_{x}$ for all $x\in X$. Then for every $i'\geq i$, the strip $\llbracket L_{i'}(x), L_{i'+1}(x) \rrbracket$ contains at most $O(1)$ points that are in some candidate paths in $G$ from points in $X$ to $s$. 
\end{claim}
\begin{proof}
    Since $X\subseteq N_{\tau}(L_i)$, then $L_{i'}(x)=L_{i'}(x')$ for all $i'\geq i$ and $x,x\in X$. By \Cref{lem:crosssection}, the width of each rectangle at level $i$ is $\Theta(2^i \epsilon)$. The point $v \in \llbracket L_i, L_{i+1} \rrbracket$ lies in the ellipse $\mathcal{E}_{x}$ for all $x\in X$. Since the ellipses have horizontal major axes and their minor axes are to the right of $s$, then the intersection these ellipses also has nonempty intersection with every strip $\llbracket L_{i'}(x), L_{i'+1}(x) \rrbracket$ at level $i' \geq i$. Consequently, the intersection of all rectangles associated with points in $X$ is nonempty in the strip $\llbracket L_{i'}(x), L_{i'+1}(x) \rrbracket$ for all levels $i' > i$. Let $\mathcal{R}_{i'}$ denote the set of rectangles associated with points in $X$ that lie in the strip $\llbracket L_{i'}(x), L_{i'+1}(x) \rrbracket$.

    The hitting set problem for $\mathcal{R}_{i'}$ reduces to a one-dimensional hitting set problem by projecting all points and rectangles onto a vertical line. Each rectangle in $\mathcal{R}_{i'}$ then becomes an interval. Let $\mathcal{I} =\{I_1, I_2, \ldots, I_h\}$ be the set of intervals corresponding to rectangles in $\mathcal{R}_{i'}$. Since there is one point in the intersection $\bigcap_{1 \leq j \leq h} I_j$, the union of all intervals forms a segment of length $\Theta(2^{i'} \epsilon)$. Moreover, by \Cref{lem:crosssection}, each interval $I_j$ has length $\Theta(2^{i'} \epsilon)$. Therefore, any minimal hitting set for the intervals in $\mathcal{I}$ has constant size.
\end{proof}

Using a similar proof, we obtain the following claim:
\begin{claim}
    \label{clm:bounded-1-deg}
    For every vertex $u\in P$ and every index $i$, there are $O(1)$ points $v\in P$ such that $\iter_p(u) = i$ and $\iter_p(v) = i - 1$ for some $p\in N_\tau$.
\end{claim}

Since there are $O(\log \eps^{-1})$ levels, this gives a bound on the maximum vertex degree in $G$.

\begin{corollary}
    \label{cor:all-level-deg}
    For every vertex $u\in P$, there are $O(\log(\eps^{-1}))$ points $v\in P$ such that
    $\iter_p(u) = \iter_p(v) - 1$ for some $p\in N_\tau$.
\end{corollary}

For any straight-line graph $M$ drawn in the plane and any continuous geometric region $S$, we define $\wts(M \cap S)$ to be the weight of the geometric intersection between $M$ and $S$. For each edge $e$, we only count the portion of its length that lies within $S$. 

\begin{claim}
    \label{clm:weight_size_bound}
    Given a strip $\llbracket L_i, L_{i+1} \rrbracket$ with $L_i \in \mathcal{L}_i$ and $L_{i+1} \in \mathcal{L}_{i+1}$, assume that $S$ is a minimum hitting set for the rectangles associated with $\llbracket L_i, L_{i+1} \rrbracket$. Then, we have 
    $\wts(G \cap \llbracket L_i, L_{i+1} \rrbracket) \geq 4^{i + 1}\eps \cdot |S|$ and 
    $\wts(\OPT \cap \llbracket L_i, L_{i+1} \rrbracket) \geq \Omega(4^{i + 1}\eps \cdot |S|)$.
\end{claim}

\begin{proof}
    For each point $u \in S$, $G$ contains two edges $(u, v_u), (w_u, u)$ with $\iter_p(v_u) < i < \iter_p(w_u)$ for some $p \in N_\tau$. Hence,
    \begin{equation*}
        \wts(G \cap \llbracket L_i, L_{i+1} \rrbracket) \geq \sum_{u \in S}((\underbrace{\wts((u, v_u) \cap \llbracket L_i, L_{i+1} \rrbracket) + \wts((u, w_u) \cap \llbracket L_i, L_{i+1} \rrbracket}_{\geq 4^{i + 1}\eps})) \geq 4^{i + 1}\eps \cdot |S|.
    \end{equation*}

    By \Cref{lem:crosssection}, the width of each rectangle at level $i$ is $\Theta(2^i \epsilon)$. Assume that the width of each rectangle is within the a range $[c_1 \cdot 2^i \epsilon, c_2 \cdot 2^i \epsilon]$, where $0 < c_1 < c_2$.
    Assume further that $S=\{u_1, u_2, \ldots ,u_h\}$, where the points are labeled by increasing $y$-coordinates. Observe that for any three consecutive points, $u_j$, $u_{j + 1}$, and $u_{j + 2}$, we have 
    \begin{equation}\label{eq:triple}
        |y(u_j) - y(u_{j + 2})| \geq c_1\cdot 2^i \epsilon ,
    \end{equation}
    otherwise $u_j$ or $u_{j + 2}$ would hit every rectangle that $u_{j + 1}$ hits
    and $u_{j + 1}$ would be redundant, contradicting the minimality of $S$. Let $l_0 = 4\lceil c_2/c_1 \rceil$. Summation of \Cref{eq:triple} over $j=1,\ldots , l_0-2$ yields
    \begin{equation}
        |y(u_{j + l_0}) - y(u_j)| 
        \geq  \frac{l_0\cdot c_12^i \epsilon }{2} 
        \geq 2c_22^i \epsilon.
    \end{equation}
    Thus, any pair of rectangles at level $i$ hit by $u_j$ and $u_{j + l_0}$, resp., are disjoint. Consider the set $S' = \{u_{1 + j \cdot l_0}\in S : j \leq \lfloor\frac{h - 1}{l_0}\rfloor\}$; and let $\mathcal{R}_{S'}$ be the set of $|S'|$ disjoint rectangles hit points in $S'$. 
   Recall that $\OPT$ contain a path crossing each rectangle on $\llbracket L_i, L_{i+1} \rrbracket$. Since the rectangles in $\mathcal{R}_{S'}$ are disjoint, then $\OPT$ contains disjoint paths crossing them, which implies that $\wts(\OPT \cap \llbracket L_i, L_{i+1} \rrbracket) \geq 4^{i + 1}\eps \cdot |S'| = \Theta(4^{i + 1}\eps \cdot |S|)$.
\end{proof}

We will bound the weight of $G$ in terms of the weight of $\OPT$.
\begin{lemma}
    \label{lem:bound_weight}
    $\wts(G) = O\left(\log(\eps^{-1}) \cdot \wts(\OPT)\right)$.
\end{lemma}

\begin{proof}  
Recall that $G=(V,E)$ is the union of candidate paths $(p, s_0, s_1, \ldots ,s_{h}, s)$ for all $p\in N_\tau$. We distinguish between two types of edges in a candidate path: an edge $(s_j,s_{j+1})$ is \emph{local} if $\iter_p(s_{j+1})=\iter_p(s_j)+1$, and \emph{crossing} otherwise. Based on this, we partition the edges of $G$ as follows:
Let $E_{\rm loc}$ be the set of edges $(u,v)$ in $E$ such that $(u,v)$ is a local edge is a candidate path for some $p\in N_\tau$, and let $E_{\rm cr}=E\setminus E_{\rm loc}$.
We derive upper bounds for the weight of local and crossing edges separately. 

\paragraph{Local edges.} For local edges, we argue similarly to the weight analysis in \Cref{sec:sector}. Specifically, let $\llbracket L_{i-1}, L_{i}\rrbracket$ and $\llbracket L_i, L_{i+1}\rrbracket$ be two consecutive strips such that $L_{i-1}=L_{i-1}(p)$, $L_i=L_i(p)$ and $L_{i+1}=L_{i+1}(p)$ for some $p\in N_\tau$. Observe that there are at most $O(1)$ such $L_{i - 1}$ given $L_i$ and $L_{i + 1}$. Let $E_{\rm loc}(L_{i-1},L_i,L_{i+1})$ be the set of edges $(u,v)$ such that $u\in \llbracket L_{i-1}, L_{i}\rrbracket$ and  $v\in \llbracket L_{i}, L_{i+1}\rrbracket$. 
We show that
\begin{equation}\label{eq:local}
    \wts(E_{\rm loc}(L_{i-1},L_i,L_{i+1})) \leq \wts(\OPT \cap \llbracket L_{i}, L_{i + 1}\rrbracket).
\end{equation}
Let $S$ be the hitting set founded for the rectangles in the strip $\llbracket L_{i}, L_{i + 1}\rrbracket$. For each vertex $u \in S$, let $E_{\llbracket L_{i-1}, L_i\rrbracket}(u)$ be the set of local edges connecting $u$ to a vertex in $\llbracket L_{i-1}, L_i\rrbracket$. By \Cref{clm:bounded-1-deg}, $|E_{\llbracket L_{i-1}, L_i\rrbracket}(u)| = O(1)$. Furthermore, each edge in $E_{\rm loc}(L_{i-1},L_i,L_{i+1})$ has length $O(4^{i + 1}\eps)$. Thus, 
\begin{equation*}
    \begin{split}
        \wts(E_{\rm loc}(L_{i-1},L_i,L_{i+1})) &\leq \sum_{u \in S}\sum_{e \in E_{\llbracket L_{i-1}, L_i\rrbracket}(u)} \wts(e)\\
        &\leq |S| \cdot O(4^{i + 1}\eps) \leq \wts(\OPT \cap \llbracket L_{i}, L_{i + 1}\rrbracket),
    \end{split}
\end{equation*}
    by \Cref{clm:weight_size_bound}.
For each line $L_i$, let $\mathcal{B}_i = \{L_{i - 1}(p) : p \in P_{\tau}, L_i(P) = L_i\}$. Summation of \Cref{eq:local} over all strips $\llbracket L_{i}, L_{i + 1}\rrbracket$ yields: 
\begin{equation*}
    \begin{split}
        \wts(E_{\rm loc})&\leq \sum_{0 \leq i \leq k - 1}\sum_{L_i \in \mathcal{L}_i}\sum_{L_{i - 1} \in \mathcal{B}_i}  \wts(E_{\rm loc}(L_{i-1},L_i,L_{i+1}))\\
        &\leq \sum_{0 \leq i \leq k - 1}\sum_{L_i \in \mathcal{L}_i}\wts(\OPT \cap \llbracket L_{i}, L_{i + 1}\rrbracket) \\
        &\leq \sum_{0 \leq i \leq k - 1}\wts(\OPT) = O(\log{\eps^{-1}}\cdot \opt).
    \end{split}
\end{equation*}

\paragraph{Crossing edges.}
We bound $\wts(E_{\rm cr})$ in two steps: First we choose a subset $E^*_{\rm cr}\subset E_{\rm cr}$ of \emph{representative} crossing edges and show that $\wts(E_{\rm cr})\leq  O(\wts(E_{\rm cr}^*))$. Then we \emph{charge} each representative edge $e$ to an edge of $\OPT$ of weight $\Omega(\wts(e))$ and show that each edge $f$ of $\OPT$ receives $O(\log\eps^{-1}\cdot \wts(f))$ charges. The charging scheme readily implies $\wts(E_{\rm cr}^*)\leq  O(\log \eps^{-1}\cdot \opt)$. Overall, we obtain $\wts(E_{\rm cr})\leq O(\wts(E_{\rm cr}^*))\leq O(\log \eps^{-1}\cdot \opt)$.

Consider a strip $\llbracket L_i, L_{i+1}\rrbracket$ where $L_i\in \mathcal{L}_i$ and $L_{i+1}\in \mathcal{L}_{i+1}$. Let $E_{\rm cr} \llbracket L_i, L_{i+1}\rrbracket$ 
be the set of edges $(u,v)\in E_{\rm cr}$ such that $(u,v)$ is an edge of a candidate path from $p$ to $s$ such that $L_i=L_i(p)$ and $L_{i+1}=L_{i+1}(p)$ and $\iter_p(u)=i$. In particular, this implies $u\in \llbracket L_i, L_{i+1}\rrbracket$.
Let $S$ be the minimum hitting set found for the rectangles associated with $\llbracket L_i, L_{i+1}\rrbracket$. For each $u\in S$, let $E_{\rm cr}(u)$ be the set of edges $(u,v)$ in $E_{\rm cr}$; and if $E_{\rm cr}(u)\neq \emptyset$, then let $(u,v*_u)$ be a longest edge in  $E_{\rm cr}(u)$. We claim that for every $u\in S$,
\begin{equation}\label{eq:longest}
    \wts(E_{\rm cr}(u))\leq O(d(u,v^*_u)) .
\end{equation}
We may assume w.l.o.g.\ that $E_{\rm cr}(u)\neq \emptyset$. By \Cref{cl:perlevel}, $u$ is adjacent to $O(1)$ vertices in each level. 
By \Cref{obs:spacing}, the width if a strip at level $j$ is $\Theta(4^j\eps)$.
The summation of geometric series implies that the total width of strips at consecutive levels equals the width of the strip of the largest level, up to constant factors. 
Consequently, if $(u,v)\in E_{\rm cr}(u)$ and $v\in \llbracket L_j(p), L_{j+1}(p)\rrbracket$, then $d(u,v)=\Theta(2^j\eps)$. 
Applying geometric series again, we obtain \Cref{eq:longest}.

For each $u\in S$, the edge $(u,v_u^*)$ is part of some candidate path 
starting from a point $p_u\in  P\llbracket L_i, L_{i+1}\rrbracket$. 
Let $R(u):=B_i(p_u)$ is the minimum axis-aligned rectangle enclosing   $\mathcal{E}_{p_u}\cap \llbracket L_i, L_{i+1}\rrbracket$. 

Consider the set of rectangles $\mathcal{R}=\{R(u) : u\in S\}$. 
In the proof of \Cref{clm:weight_size_bound}, we have shown that every rectangle in $\mathcal{R}$ intersects $O(1)$ other rectangles. Therefore $\mathcal{R}$ can be partitioned into $O(1)$ sets of pairwise disjoint rectangles. 
Let $\mathcal{R}^*$ be the subset of disjoint rectangles 
that maximizes $\sum_{R(u)\in \mathcal{R}^*} d(u,v_u^*)$.
Then we have 
$\sum_{R(u)\in \mathcal{R}^*} d(u,v_u^*) 
\geq \Omega\left(\sum_{R(u)\in \mathcal{R}} d(u,v_u^*)\right)$. 
Combined with \Cref{eq:longest}, this gives 
\begin{equation}\label{eq:star}
    \wts\big( E_{\rm cr} \llbracket L_i, L_{i+1}\rrbracket\big) 
    \leq O\left(\sum_{R(u)\in \mathcal{R}^*} d(u,v_u^*)\right).
\end{equation}
We define the set of \emph{representative edges} in $E_{\rm cr} \llbracket L_i, L_{i+1}\rrbracket$ to be the edges $(u,v_u^*)$ where $R(u)\in \mathcal{R}^*$. 
Now \Cref{eq:star} states that the weight of crossing edges associated with any strip $\llbracket L_i, L_{i+1}\rrbracket$ is upper bounded by constant times the weight of the representative edges in that strip.

Next, we \emph{charge} each representative edge $(u,v_u^*)$ of $E_{\rm cr} \llbracket L_i, L_{i+1}\rrbracket$ to an edge of $\OPT$. Assume that $v_u^*\in \llbracket L_j(p_u), L_{j+1}(p_u)\rrbracket$. The ellipse $\mathcal{E}_{p_u}$ does not contain any point in the strip 
$\llbracket L_{i+1}(p_u), L_j(p_u)\rrbracket$: Indeed, if $\mathcal{E}_{p_u}$ 
contains a point in this strip, then it would also contain a point in the strip 
$\llbracket L_{i'}(p_u), L_{i'+1}(p_u)\rrbracket$ for some $i<i'<j$, and so the candidate path starting from $p_u$ would also contain a vertex in that strips, contradicting our assumption that the candidate path starting from $p_u$ crosses $\llbracket L_{i'}(p_u), L_{i'+1}(p_u))\rrbracket$ for all $i<i'<j$. This implies that the path of $\OPT$ from $p_u$ to $s$ does not have any vertices in the strip $\llbracket L_{i+1}(p_u), L_j(p_u)\rrbracket$; in particular, this path in $\OPT$ contains an edge $(a,b)$ that crosses  the strip $\llbracket L_{i+1}(p_u), L_j(p_u)\rrbracket$.
We charge the weight of the representative edge $(u,v_u^*)$ to such an edge $(a,b)$ of $\OPT$. Note that $\wts(a,b)\geq \Omega(\wts(u,v_u^*))$. 

Furthermore, \Cref{cor:spacing} implies 
that $\wts(p_u, L_{i+1}(p_u)\leq \frac{32}{3}\cdot 4^{i+1}\eps$. 
Since $(a,b)$ is an edge of the path from $p_u$ to $s$ in $\OPT$, 
but $a$ is to the left of $L_{i+1}$, then 
\begin{equation}\label{eq:near}
    \wts(a, L_{i+1}(p_u))\leq \wts(p_u, L_{i+1})\leq \frac{32}{3}\cdot 4^{i+1}\eps.
\end{equation}

It remains to show that every edge $(a,b)$ of $\OPT$ receives $O(\wts(a,b)\cdot \log \eps^{-1})$ charges. Consider an edge $(a,b)$ of $\OPT$. 
At each level $i$, point $a$ lies in $O(1)$ distinct strips $\llbracket L_i(p), L_{i+1}(p)\rrbracket$ for some $p\in N_\tau$. Let $\llbracket L_i, L_{i+1}\rrbracket$ be one of these strips. 
Suppose that $(a,b)$ receives charges from a representative edge $(u,v_u^*)$ where $\iter_{p(u)}(u)=i$. Then $u\in \llbracket L_i(p_u), L_{i+1}(p_u)\rrbracket$, where  $L_i(p(u))\in \mathcal{L}_i$ and $L_{i+1}(p(u))\in \mathcal{L}_{i+1}$. Since the distance between two consecutive vertical lines in $\mathcal{L}_i$ is $4^i\cdot \eps$, then \Cref{eq:near} implies that $u$ may be located in $O(1)$ possible strips at level $i$. In the strip $\llbracket L_i(p_u), L_{i+1}(p_u)\rrbracket$, the representative edges correspond to disjoint rectangles $\mathcal{R}^*$, and $\OPT\cap \llbracket L_i(p_u), L_{i+1}(p_u)\rrbracket \subset R(u)$. Therefore, $(a,b)$ receives charges from at most one edge in each strip. We have shown that $(a,b)$ receives charges from $O(1)$ edges at each level $i$. 
Over all $O(\log \eps^{-1})$ levels, $(a,b)$ may receive charges from $O(\log \eps^{-1})$ edges in $E_{\rm cr})$. As noted above, whenever we charge $\wts(u,v_u^*)$ to $(a,b)$, 
we have $\wts(a,b)\geq \Omega(\wts(u,v_u^*))$. Consequently, the total amount of charges received by $(a,b)$ is $O(\wts(a,b)\cdot \log \eps^{-1})$. 
\end{proof}

To complete the proof, we show that: 

\begin{lemma}\label{lem:skipping}
    $\wts(G') = O(\log \eps^{-1}\cdot \wts(G) )$.
\end{lemma}

We first upper-bound on the total weight of edges incident to a vertex $u$ from one level below.
\begin{claim}\label{clm:bdd_max_len}
    For each point $u$ in $P$, let $i_{\rm max}(u)$ be the largest level $i$ such that $u$ is in a minimum hitting set of $\mathcal{R}_i(L_i, L_{i + 1})$ for some $L_i \in \mathcal{L}_i$ and $L_{i + 1} \in \mathcal{L}_{i + 1}$. If there is no such $i$, let $i_{\rm max}(u) = -\infty$. Then, $\sum_{u \in P}4^{i_{\rm max}(u)}\eps =O(\wts(\OPT)\cdot \log(\eps^{-1}))$.
\end{claim}

\begin{proof}
    For every strip $\llbracket L_i, L_{i + 1}\rrbracket$, let $S(L_i, L_{i + 1})$ be the minimum hitting set of the set of rectangles lying in $\llbracket L_i, L_{i + 1}\rrbracket$. Using \Cref{lem:crosssection} and \Cref{clm:weight_size_bound}, we obtain
    \begin{equation}
        \sum_{0 \leq i < k}\sum_{L_i \in \mathcal{L}_i}|S(L_i, L_{i + 1})|\cdot 4^i\eps 
        \leq  \sum_{0 \leq i < k}\sum_{L_i \in \mathcal{L}_i} O\left(\wts(\OPT \cap \llbracket L_i, L_{i + 1}\rrbracket)\right) 
        = O\left(\log\eps^{-1} \cdot \wts(\OPT)\right).
    \end{equation}
    On the other hand, observe that $\sum_{u \in P}4^{i_{\rm max}(u)}\eps \leq \sum_{0 \leq i < k}\sum_{L_i \in \mathcal{L}_i}|S(L_i, L_{i + 1})|\cdot 4^i\eps$, as each term on the left hand side contribute to at most one on the right hand side. Therefore, we obtain $\sum_{u \in P}4^{i_{\rm max}(u)}\eps = O\left(\log\eps^{-1}\cdot \wts(\OPT)\right)$.
\end{proof}

\begin{proof}[Proof of \Cref{lem:skipping}]
    We show that the total weight of edges in $E(G')\setminus E(G)$ is $O(\log{\eps^{-1}} \cdot \wts(G))$. Note that every edge in $E(G')\setminus E(G)$ is a shortcut edge $(s_j,s_{j+2})$ in an approximate path starting from some point $p\in N_\tau$, where the candidate path starting from $p$ contained $(s_j,s_{j},s_{j+1})$ with $\iter_p(s_{j+1})=\iter_p(s_j)+1$, and we removed $s_{j+1}$ from the path.  
    A triple $(u, v, w)$ is a \EMPH{candidate triple} of $p$ if $(u, v, w)$ is a subpath of the candidate path from $p$ to $s$ in $G$ and $\iter_p(v) = \iter_p(u) + 1$.
    
    Fix a vertex $w$. 
    Let $N(u)$ be the set of neighbors of $u$ in $G$. Let $Z(w)$ be the set of pairs $(u, v)$ such that $(u, v, w)$ is a candidate triple; $Z_1(w)$ the set of pairs $(u, v)$ such that $(u, v, w)$ is a candidate triple of some point $p\in P_\tau$ and $\iter_p(w) = \iter_p(v) + 1$ (hence $\iter_p(w)=\iter_p(u) + 2$); and $Z_{>1}(w) = Z(w) \setminus Z_1(w)$. Thus,
    \begin{equation}
        \label{eq:bdd_change}
        \wts(G') - \wts(G) \leq \sum_{w \in P}\left( \sum_{(u, v) \in Z_1(w)}\wts(u, w) + \sum_{(u, v) \in Z_{>1}(w)}\wts(u, w)\right)
    \end{equation}

    Consider $Z_1(w)$. At each level $i$, there are constantly many strips $\llbracket L_i, L_{i + 1}\rrbracket$ that $w$ belongs to. Furthermore, for each strip, there are constantly many pairs $(u, v)$ such that there exists a point $p$ where $L_i(p) = L_i$ and $(u, v, w)$ is a candidate triple of $p$ in $Z_1(w)$ by \Cref{clm:bounded-1-deg}. Let $L(u)$ be the set of levels $i$ such that there exists some strip $\llbracket L_i, L_{i + 1}\rrbracket$ containing $u$ as a point in the hitting set. Then,     
    \begin{equation}
        \label{eq:bdd_short_subpath}
        \sum_{(u, v) \in Z_1(w)}\wts(u, w) 
        \leq \sum_{i \in L(w)}O(4^{i + 1}\eps) 
        = O(4^{\max L(w)}\eps). 
    \end{equation}

    Consider $Z_{>1}(w)$. By \Cref{cor:all-level-deg}, for a fixed $v$, there are $O(\log(\eps^{-1}))$ pairs $(u, v) \in Z_{>1}(w)$.
    For each pair $(u, v) \in Z_{>1}(w)$, let $p$ be the point that $(u, v, w)$ is the candidate triple of and $\iter_p(w) > \iter_p(v) + 1$. Assume that $u \in \llbracket L_i(p), L_{i + 1}(p)\rrbracket$ for some $i$.
    Then, $\wts(u, v) \leq 4^{i + 1}\eps = \Theta(\wts(w, v))$
    Then, by the triangle inequality, we have:
    \begin{equation}
    \begin{split}
        \label{eq:bdd_long_subpath}
        \sum_{(u, v) \in Z_{>1}(w)}\wts(u, w) &\leq \sum_{(u, v) \in Z_{>1}(w)} (\wts(u, v) + \wts(v, w)) \leq \sum_{(u, v) \in Z_{>1}(w)} \Theta(\wts(w, v))\\
        &\leq O(\log\eps^{-1}) \cdot \sum_{v \in N(w)}\wts(w, v).
    \end{split}
    \end{equation}

    The combination of \Cref{eq:bdd_change,eq:bdd_short_subpath,eq:bdd_long_subpath} yields a bound for the total weight of $G'$:
    \begin{equation*}
        \sum_{w \in P} \left(O(4^{\max L(w)}\eps) + O(\log\eps^{-1}) \cdot \sum_{v \in N(w)}\wts(w, v) \right) = O\left(\log\eps^{-1} \cdot \wts(G)\right),
    \end{equation*} 
    by \Cref{clm:bdd_max_len}. 
\end{proof}
    
\paragraph{Running time.}
The running time analysis is similar to the Steiner version. First, we reduce the problem to computing $\epsilon$-cnets in a tile $\tau$ in $O(n \log n)$ time. Let $n_\tau$ be the number of points in $N_\tau$. Since $N_\tau$ is an $\epsilon$-cnet, we have $n_\tau = |N_\tau| = O(\epsilon^{-2})$. We show that the total running time is $O({\rm poly}\log \eps^{-1} \cdot n \log n )$.

We first compute the set of lines associated with each point in $N_\tau$, similarly to the Steiner version. Observe that the hitting set for rectangles in any strip can be reduced to the hitting set problem for intervals, which can be solved exactly using a greedy algorithm in $O(m \log m)$ time, where $m$ is the total number of intervals plus candidate points. This is the only difference between the Steiner and the non-Steiner algorithms. 

The total number of intervals (or rectangles) is $O(n_\tau \log \epsilon^{-1})$, and each candidate point belongs to at most $\log \epsilon^{-1}$ strip. Each point in $P$ belongs to at most $O(\log \eps^{-1})$ strips in $\tau$. Therefore, the total time needed to solve all the hitting set instances over all tiles is $\sum_\tau O(n_\tau \log n) + O(\log \eps^{-1}) \cdot n\log{n} = O(\log \eps^{-1}) \cdot n\log{n}$.

The graph $G$ for a single tile $\tau$ has at most $O(n_\tau \cdot \log \epsilon^{-1})$ vertices and edges, so building $G$ takes at most $O(n_\tau \cdot \log \epsilon^{-1})$ time. Adding edges and constructing $G'$ can be done in time linear in $|V(G)| + |E(G)|$. The shortest path tree $T$ rooted at $s$ can be computed in
\[
O(n_\tau \cdot \log \epsilon^{-1} \cdot \log(n_\tau \cdot \log \epsilon^{-1})) = O(n_\tau \log n_\tau \cdot \mathrm{polylog}(\epsilon^{-1})).
\]

Summing over all tiles, the total running time of the graph computation is
\begin{equation*}
    \sum_{\tau} O(n_\tau \log n_\tau \cdot \mathrm{polylog}(\epsilon^{-1})) = O(n \log n \cdot \mathrm{polylog}(\epsilon^{-1})).
\end{equation*}

\section{Generalization to Higher Dimensions} 
\label{sec:dspace}

Our algorithms generalize to Euclidean $d$-space for any constant $d\geq 2$. We briefly discuss the necessary adjustments in $d$-space. The bottleneck in the running time is the minimum hitting set problem. For $(1+\eps)$-STs in the plane, we used minimum hitting sets for $N$ intervals on a line, which can be computed exactly in $O(N\log N)$ time. 
However, for constructing STs in $\mathbb{R}^d$, we need minimum hitting sets for balls (or possibly cubes) of comparable sizes in $\mathbb{R}^{d-1}$.  

\paragraph{Hitting set versus piercing set.}
There are two versions of minimum hitting sets: \emph{discrete} and \emph{continuous}. We use the discrete version for our non-Steiner SLT construction and the continuous version (also known as \emph{piercing set}) for our Steiner SLT. In the \emph{discrete} setting, we are given a set $\mathcal{R}$ of ranges and a finite point set $P$ in $\mathbb{R}^d$, and a \emph{hitting set} is a subset $H\subset P$ such that every range in $\mathcal{R}$ contains at least one point in $H$. In the \emph{continuous setting}, we are given a set $\mathcal{R}$ of ranges in $\mathbb{R}^d$, and a \emph{piercing set} is a subset $H\subset \mathbb{R}^d$ such that every range in $\mathcal{R}$ contains at least one point in $H$.

Finding the \emph{exact} minimum hitting set for Euclidean balls is NP-hard~\cite{FowlerPT81}, even for unit disks in the plane~\cite{HochbaumM87}.
However, an approximate hitting set is sufficient for our purposes. 
Mustafa et al.~\cite{MustafaR10} gave a PTAS for the hitting set problem for disks in $\mathbb{R}^2$, and Agarwal and Pan~\cite{AgarwalP20} gave a randomized $O(1)$-approximation algorithm in near-linear time (their algorithm for the dual set cover problem, with the same approximation ratio, has since been derandomized by \cite{ChanH20}). 
Agarwal and Pan~\cite{AgarwalP20} also gave a randomized $O(\log \log \opt)$-approximation for hyper-rectangles in $\mathbb{R}^d$ for $d=2$ and $d=3$. 
In higher dimensions, however, we are left with the greedy $O(\log N)$-approximation algorithm. Importantly, all instances of the hitting set problem in our constructions,
the ranges are Euclidean balls in $\mathbb{R}^{d-1}$ generated by a $\eps$-cnet $N_i$ in a tile $\tau_i$. Standard volume argument shows that $|N_i|=O(\eps^{-d})$ for every $i$. 
The arrangement of $O(\eps^{-d})$ Euclidean balls in $\mathbb{R}^{d-1}$ generates 
$O(\eps^{-d(d-1)})$ cells~\cite{HS17}. Points in the same cell hit the same balls,
consequently we can reduce the size of the ground set from $n$ to $N=O(\eps^{-d(d-1)})$.
An $O(\log N)$-approximate hitting set incurs only an $O(\log \eps^{-1})$ increase in the weight of the resulting SLT, and does not impact the stretch analysis. 

Finding the exact minimum piercing set problem is also NP-hard, already for unit squares in the plane~\cite{FowlerPT81}. However, the minimum piercing set can be approximated efficiently. There is an $O(N\log N)$-time $O(1)$-approximation algorithms for the minimum piercing set of $N$ fat objects (including Euclidean balls) in $\mathbb{R}^d$ for any constant dimension $d\in \mathbb{N}$~\cite{EfratKNS00,MaratheBHRR95}; and a PTAS is also available~\cite{Chan03,ErlebachJS05}. For axis-aligned hyperrectangles in $\mathbb{R}^d$, in general, only an $O(\log \log \opt)$-approximation is known~\cite{AgarwalHRS24}.

\paragraph{Shallow tree constructions in $d$-space.}
We can now briefly go over the necessary adjustments to extend out algorithms to $\mathbb{R}^d$, for $d\geq 3$. Instead of ellipses $\mathcal{E}_{ps}$, we use ellipsoids with foci $p$ and $s$, and major axis $(1+\eps)\cdot d(p,s)$. The reduction to a $\eps$-cnet in a tile (\Cref{sec:reduction}) easily generalizes to $\mathbb{R}^d$: The only difference is that we cannot tile $\mathbb{R}^d$ be trapezoids: Instead we can \emph{cover} $\mathbb{R}^d$ with cones with apex angle $\sqrt{\eps})$, and obtain a \emph{covering} of $\mathbb{R}^d$ with trapezoids. We can still assign $n$ points to trapezoids in $O(n)$ time. 
In each tile $\tau_i$,  standard volume argument shows that the size of an $\eps$-cnet is $O(\eps^{-d})$. While point location data structures are inefficient in higher dimensions, we can assign the $n_i$ points to $O(\eps^{-d})$ clusters by brute force in $O(n_i\cdot \eps^{-d})$ time. In each cluster $C_a$, $a\in N$, we can still compute a (Steiner) $O(1)$-spanner in $O(|C_a|\log |C_a|)$ time~\cite[Theorem~1.1]{LeS23}. Thus, in a tile $\tau_i$ with $n_i$ points, the reduction to an $\eps$-cnet takes $O(n_i(\log n_i + \eps^{-d}))$ time.

Our (Steiner and non-Steiner) algorithms for a $\eps$-cnet $P_i$ in a tile $\tau_i$ can easily be adapted to $d$-dimensions. We replace the parallel lines $L_i(p)$, for $i=1,\ldots , k-1$, for $k=O(\log \eps^{-1})$, with parallel hyperplanes, and the intervals $\mathcal{E}_p\cap L_i(p)$ with balls of comparable radii in $\mathbb{R}^{d-1}$. For $d=3$, we can use a randomized $O(1)$-approximation of the minimum hitting sets for disks in 2D~\cite{AgarwalP20}; and for $d\geq 4$, we use a $O(\log N)=O(\log \eps^{-1})$-approximate greedy hitting sets. In our algorithms, each invocation of the hitting set problem involves $N=O(n_i)$ balls in a $(d-1)$-dimensional hyperplane, where $n_i=O(\eps^{-d})$. Using an approximate hitting sets incurs an $O(\log N)=O(d\log \eps^{-1})$ factor increase in the number of hitting points, hence the weight of the ultimate ST.
An $O(1)$-approximate minimum piercing set can be computed in $O(n_i\log n_i)$ time for $n_i$ balls in $\mathbb{R}^{d-1}$, for any constant $d\geq 3$~\cite{EfratKNS00,MaratheBHRR95}. The use of approximate piercing sets increases the weight of our Steiner SLT by only a constant factor. In both cases (Steiner and non-Steiner), our root-stretch analysis works for any hitting (resp., piercing) set, and so the root stretch is not impacted by the use of suboptimal hitting (resp., piercing) sets.   

\section{Lower Bounds}
\label{sec:LB}

\subsection{Existentially Optimal Algorithms are Not Instance-Optimal}
\label{ssec:existential}

Recall (\Cref{sec:intro}) that for any $n$ points in Euclidean plane and any $\eps>0$, Awerbuch et al.~\cite{ABP90,ABP91} and Kuller et al.~\cite{KRY95} constructed an $(1+\eps,O(\frac{1}{\eps}))$-SLT and Solomon~\cite{Solomon15} constructed a Steiner $(1+\eps,O(\frac{1}{\eps}))$-SLT. Both bounds are existentially optimal, but they are not necessarily instance-optimal. In this section, we construct instances (i.e., point sets in the plane) for which these algorithms perform poorly: They return $(1+\eps)$-ST (resp., Steiner $(1+\eps)$-ST) of weight significantly larger than optimum. 

First (\Cref{obs:LB}) we present instances for which the approximation ratio attained by the algorithms in \cite{ABP90,ABP91,KRY95} and \cite{Solomon15} are $\Omega(\frac{1}{\eps})$ and $\Omega(\sqrt{1/\eps})$, resp., and our reduction to a centered $\eps$-net cf.~\Cref{sec:reduction}) is sufficient to find an $O(1)$-approximation. 
Then we construct point sets that are already centered $\eps$-nets in a tile (\Cref{thm:LB}), and yet our algorithm vastly outperforms the algorithms in~\cite{KRY95,Solomon15}. 

\paragraph{Existentially Optimal Algorithms.} We start with a brief description of existentially optimal algorithms for a set $P$ of $n$ points in the plane. The algorithm by Khuller et al.~\cite{KRY95} visits the points in a DSF traversal of an $\mst(P)$, and maintains a \emph{current tree} $T$. Initially, $T:=\mst(P)$, and $T$ is represented by the parent function $\pi(v)$ for every vertex $v\in V\setminus \{s\}$. For each vertex $v$ in encountered in the traversal, the algorithm checks whether $d_T(v,s)\leq (1+\eps)d(v,s)$: If not, then it replaces the edge $v\pi(v)$ of $T$ with the new edge $vs$; and sets $\pi(v):=s$.

The algorithm by Awerbuch et al.~\cite{ABP90,ABP91} constructs a Hamiltonian path $H$ that visits the points in an order determined by a DFS traversal of $\mst(P)$, where 
$\wts(\mst(P))\leq \wts(H)\leq 2\, \wts(\mst(P))$. They greedily break $H$ into subpaths $H_i$ such that $\wts(H_i)=\eps\cdot \min_{p\in H_i}(p,s)$; and finally add an edge between $s$ and the closest vertex in each subpath. 

The algorithm by Solomon~\cite{Solomon15} proceeds similarly to that of Awerbuch et al.~\cite{ABP90,ABP91}: It constructs a Hamiltonian path $H$ with 
$\wts(\mst(P))\leq \wts(H)\leq 2\, \wts(\mst(P))$, but it breaks $H$ into subpaths $H_i$ such that $\wts(H_i)=\sqrt{\eps}\cdot \min_{p\in H_i}(p,s)$. Between $s$ and each subpath, it adds a Steiner tree of weight $O(\min_{p\in H_i}(p,s))$.

\begin{observation}\label{obs:LB}
    For every $\eps>0$, there exists a set $P\subseteq \mathbb{R}^2$ and a source $s\in P$ such that the minimum weight of an $(1+\eps)$-ST (resp., Steiner $(1+\eps$)-ST) is 
    $O(1)\cdot w(\mst(P)),$
    but any previous algorithm in~\cite{ABP90,ABP91,KRY95}
    returns an $(1+\eps)$-ST of weight $\Omega(\frac{1}{\eps})\cdot w(\mst(P))$,
    and the algorithm in~\cite{Solomon15}
    returns a Steiner $(1+\eps)$-ST of weight $\Omega(\frac{1}{\sqrt{\eps}})\cdot w(\mst(P))$.
\end{observation}
\begin{proof}
  We describe the point set $P$ for a given $\eps>0$ and a parameter $k\in \mathbb{N}$; see \Cref{fig:easyLB}. The source is $s=(2,0)$ on the $x$-axis, all other points lie in the unit cube $[0,1]^2$. Specifically,  the points lie on the the bottom side of $[0,1]^2$, and in $k$ vertical lines. Place points on these lines so that $\mst(P)$ consists the line segment between the $s$ and the origin; and the intersection of the vertical lines with the unit cube $[0,1]^2$. 
 That is, $\mst(P)$ consists of a horizontal line segment of length 2 and $k$ vertical unit line segments. Consequently, we have $\wts(\mst(P))=2+k$.
\begin{figure}[htbp]
    \centering
    \includegraphics[width=.5\textwidth]{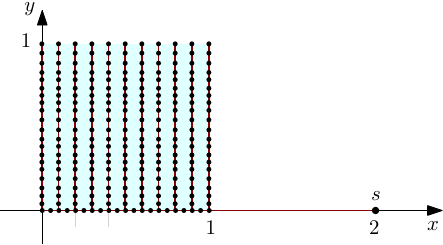}
    \caption{A source $s=(2,0)$ and points in a unit cube $[0,1]^2$.}
    \label{fig:easyLB}
\end{figure}

\paragraph{Existentially Optimal SLT Algorithms.}
In the point set $P$ defined above, every vertex $v$ on the $x$-axis satisfies $d_T(v,s)=d(v,s)$ and its parent is also on the $x$-axis. 
Therefore, the algorithm by Khuller et al.~\cite{KRY95} never deletes any edge on the $x$-axis. For every point $v\in V\setminus \{s\}$, we have $1\leq d(v,s)\leq \sqrt{5}$. It follows that an $sv$-path of length at most $(1+\eps)\cdot d(s,v)$ cannot contain vertical edges of total length $\Omega(\eps)$. Consequently, the algorithm breaks every vertical path of unit length into subpaths, each of length $O(\eps)$, and every subpath requires a new edge incident to $s$. 
The total length of the vertical paths of $\mst(P)$ is $k$: This is broken into $\Omega(k/\eps)$ subpaths, and so $\Omega(k/\eps)$ new edges incident to $s$ are added. Every edge between $s$ a point in $R$ has length at least 1, and so the total length of new edges is also $\Omega(k/\eps)$. Overall, the algorithm by Khuller et al.~\cite{KRY95} returns an $(1+\eps)$-SL of weight $\Omega(k/\eps) = \Omega(1/\eps)\cdot \wts(\mst(P))$.

The algorithm by Awerbuch et al.~\cite{ABP90,ABP91} constructs a Hamiltonian path $H$ of weight $\Theta(\wts(\mst(P)))=\Theta(k)$. Since the distance between $s$ and all other points is $\Theta(1)$, this algorithm breaks $H$ into $\Omega(k/\eps)$ subpaths, and adds an edge of weight $\Omega(1)$ for each subpath.
Consequently, it returns an $(1+\eps)$-SL of weight $\Omega(k/\eps) = \Omega(k/\eps)\cdot \wts(\mst(P))$.

The algorithm by Solomon~\cite{Solomon15}
breaks the Hamiltonian path $H$ into $\Omega(k/\sqrt{\eps})$ subpaths,
and adds a Steiner tree of weight $\Theta(1)$ for each. Consequently, it returns a Steiner $(1+\eps)$-ST of weight $\Omega(k/\sqrt{\eps}) = \Omega(\sqrt{1/\eps})\cdot \wts(\mst(P))$.

\paragraph{Minimum-Weight $(1+\eps)$-ST.}
We give an upper bound for the minimum weight of a $(1+\eps)$-ST for $P$, by reducing the problem to a centered $\frac{\eps}{2}$-net. 
Since $P\setminus \{s\}\subset [0,1]^2$, 
and $1\leq d(p,s)\leq \sqrt{5}$ for all $p\in P\setminus \{s\}$, then there exists a centered $\frac{\eps}{2}$-net $N\subseteq P$ of size $O(1/\eps^2)$. 
The weight of the $\mst$ of any $m$ points in a unit square is $O(\sqrt{m})$~\cite{Few55}, and so $\wts(\mst(N))\leq O(\frac{1}{\eps})$, hence $\wts(\mst(N\cup \{s\}))\leq O(\frac{1}{\eps})$. The algorithm in~\cite{KRY95} gives a $(1+\frac{\eps}{3})$-SL for $N\cup\{s\}$ of weight $O(\frac{1}{\eps})\cdot \wts(\mst(N\cup \{s\})) = O(1/\eps^2)$. 
By \Cref{lem:net}, $P$ admits a $(1+\eps)$-ST of weight $O(\wts(\mst(P)))+O(1/\eps^2)$.
If $k=\Omega(1/\eps^2)$, then the minimum weight of a $(1+\eps)$-ST is $O(\wts(\mst(P)))$. 
\end{proof}

\paragraph{Construction of a centered $\eps$-net in the plane.} Next we construct instances where the reduction to centered $\eps$-nets (\Cref{sec:reduction}) does not help. 

\begin{theorem}\label{thm:LB}
    For every $\eps>0$, there exists a source $s$ and a centered $\eps$-net $P\subseteq \mathbb{R}^2$ such that the minimum weight of an $(1+\eps)$-ST (resp., Steiner $(1+\eps$)-ST) is 
    $O(1)\cdot w(\mst(P)),$
    but any previous algorithm in~\cite{ABP90,ABP91,KRY95}
    returns an $(1+\eps)$-ST of weight $\Omega(\frac{1}{\eps})\cdot w(\mst(P))$,
    and the algorithm in~\cite{Solomon15}
    returns a Steiner $(1+\eps)$-ST of weight $\Omega(\frac{1}{\sqrt{\eps}})\cdot w(\mst(P))$.
\end{theorem}
\begin{proof}
We start by describing a point set $P$ for a given $\eps>0$; see \Cref{fig:LB}(left). The source is $s=(2,0)$ on the $x$-axis, all other points lie in the axis-aligned box $R=[0,1]\times [0,\sqrt{\eps}]$. Specifically, 
the points lie on the bottom side of $R$, and in vertical lines of the form $x=i\cdot 2\eps$ for all $0\leq i\leq \frac{1}{2\sqrt{\eps}}$.

Place points on these lines within $R$ so that $\mst(P)$ consists of the line segment between $s$ and the origin; and the intersection of the vertical lines with the box $R$. 
That is, $\mst(P)$ consists of a horizontal line segment of length 2 and $\Theta(\sqrt{1/\eps})$ vertical line segments each of length $\sqrt{\eps})$. Consequently, we have $\wts(\mst(P))=2+\Theta(\sqrt{\eps}\cdot \sqrt{1/\eps}) = \Theta(1)$.

\begin{figure}[htbp]
    \centering
    \includegraphics[width=\textwidth]{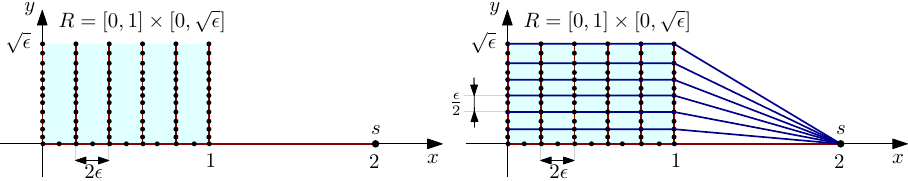}
    \caption{Left: A source $s=(2,0)$ and points in a rectangle $R=[0,1]\times [0,\sqrt{\eps}]$. Right: A graph $G$ that contains a $(1+\eps)$-SL.}
    \label{fig:LB}
\end{figure}

\paragraph{Existentially Optimal SLT Algorithms.}
In the point set $P$ defined above, every vertex $v$ on the $x$-axis satisfies $d_T(v,s)=d(v,s)$ and its parent is also on the $x$-axis. 
Therefore, the algorithm by Khuller et al.~\cite{KRY95} never deletes any edge on the $x$-axis. For every point $v\in V\setminus \{s\}$, we have $1\leq d(v,s)\leq \sqrt{5}$. It follows that an $sv$-path of length at most $(1+\eps)\cdot d(s,v)$ cannot contain vertical edges of total length $\Omega(\eps)$. Consequently, the algorithm breaks every vertical path of length $\sqrt{\eps}$ into subpaths, each of length $O(\eps)$, and every subpath requires a new edge incident to $s$. 
The total length of the vertical paths of $\mst(P)$ is $\Theta(1)$: This is broken into $\Omega(1/\eps)$ subpaths (each of length $O(\eps)$), and so $\Omega(1/\eps)$ new edges incident to $s$ are added. Every edge between $s$ a point in $R$ has length at least 1, and so the total length of new edges is also $\Omega(1/\eps)$. Overall, the algorithm by Khuller et al.~\cite{KRY95} returns a $(1+\eps)$-ST of weight $\Omega(1/\eps) = \Omega(1/\eps)\cdot \wts(\mst(P))$.

The algorithm by Awerbuch et al.~\cite{ABP90,ABP91} constructs a Hamiltonian path $H$ of wight $\Theta(\wts(\mst(P)))=\Theta(1)$. Since the distance between $s$ and all other points is $\Theta(1)$, this algorithm breaks $H$ into $\Omega(1/\eps)$ subpaths, and adds an edge of weight $\Omega(1)$ for each subpath.
Consequently, it returns an $(1+\eps)$-SL of weight $\Omega(1/\eps) = \Omega(1/\eps)\cdot \wts(\mst(P))$.

The algorithm by Solomon~\cite{Solomon15}
breaks the Hamiltonian path $H$ into $\Omega(\sqrt{1/\eps})$ subpaths,
and adds a Steiner tree of weight $\Theta(1)$ for each. Consequently, it returns a Steiner $(1+\eps)$-ST of weight $\Omega(\sqrt{1/\eps}) = \Omega(\sqrt{1/\eps})\cdot \wts(\mst(P))$.

\paragraph{Optimum $(1+\eps)$-ST.}
While we do not know the minimum weight of a $(1+\eps)$-ST for $P$, we can construct a graph that contains a path of length at most $(1+\eps)\cdot (s,v)$ for every point $p\in P$; see \Cref{fig:LB} (right): Augment $\mst(P)$ to a graph $G$ with equally spaced horizontal paths in the rectangle $R=[0,1]\times [0,\sqrt{\eps}]$ at distance $\eps/2$ apart; and connect the right endpoint of each path to $s$ by a single edge. The weight  of each path is less than $3$, so the total weight of these paths is less than $\sqrt{\eps}/(\eps/2)\cdot 3=6\cdot \sqrt{\eps}=O(\wts(\mst(P)))$, hence $\wts(G)=O(\wts(\mst(P)))$. For every vertex $p=(x(p),y(p))\in P$, we  can find a $sp$-path $\gamma_{sp}$ as follows: Follow a horizontal path to the vertical line $x=x(p)$, and then use a vertical path of weight less than $\frac{\eps}{2}$. Easy calculation shows that $\wts(\gamma_{sp})\leq (1+\eps)\, d(s,p)$: One the one hand, $d(s,p)\geq |x(s)-x(p)|=2-x(p)\geq 1$. The weight of the first edge of $\gamma_{sp}$ is at most $\sqrt{1^2+(\sqrt{\eps})^2}=\sqrt{1+\eps}< 1+\eps/2$, the weight of the horizontal path is $1-x(p)$, and the weight of the vertical path is less than $\eps/2$. Overall, $\wts(\gamma_{sp})< (2-x(p))+\eps\leq (1+\eps)\cdot (2-x(p))\leq (1+\eps)\, d(s,p)$.
\end{proof}

\subsection{Steiner SLTs in a Sector}
\label{sec:LB-sector}

One of our main results is a bi-criteria approximation for the minimum-weight Steiner $(1+\eps)$-ST for a finite point set $P\subset \mathbb{R}^2$ (\Cref{thm:main}). In \Cref{sec:sector}, we reduced the problem to a set of centered $\eps$-net in a tile $\tau\in \mathcal{T}$ (recall that $\tau$ is a region in a cone with apex $s$ and aperture $\sqrt{\eps}$, within $\Theta(1)$ distance from $s$; see \Cref{fig:spiderweb}). 

The classical lower-bound construction for this problem consists of a set $P_0$ of uniformly distributed points along a circle of unit radius centered at $s$. In this case, $\wts(\mst(P_0))\approx 1+\pi = \Theta(1)$, and any Steiner $(1+\eps)$-ST for $P_0$ has weight $\Omega(\sqrt{1/\eps})$. However, if $P\subset P_0$ is the subset of points in a cone of angle $\sqrt{\eps}$, then $\wts(\mst(P))\approx 1+\sqrt{\eps}\cdot \pi=\Theta(1)$, and there exists a Steiner $(1+\eps)$-ST of weight $O(1)=O(1)\cdot \wts(\mst(P))$~\cite{Solomon15}. This raises the question: What is the maximum lightness of a Steiner $(1+\eps)$-ST for points in a cone of aperture $\sqrt{\eps}$.

In this section, we give a lower bound on the maximum lightness of a minimum-weight Steiner $(1+\eps)$-ST for points in a cone of aperture $\sqrt{\eps}$ (\Cref{thm:LB-sector}). 
For simplicity, we place points in a rectangle.

\begin{theorem}\label{thm:LB-sector}
   For every $\tau\in \mathcal{T}$, there is a finite point set $P\subset \tau$ such that the minimum weight of a Steiner $(1+\eps)$-ST for $P\cup \{s\}$ is $\Omega(\eps^{-1/4})\cdot \wts(\mst(P\cup \{s\})$.
\end{theorem}
\begin{proof}
We start by describing a point set $P$ for a given $\eps>0$; see \Cref{fig:LB2}. The source is $s=(2,0)$ on the $x$-axis, all other points lie in the axis-aligned box $R=[0,1]\times [0,\sqrt{\eps}]$. Specifically, 
the points lie on the bottom side of $R$, and in vertical lines of the form $x=i\cdot \sqrt{\eps}$ for all $0\leq i\leq \sqrt{1/\eps}$.

\begin{figure}[htbp]
    \centering
    \includegraphics[width=.5\textwidth]{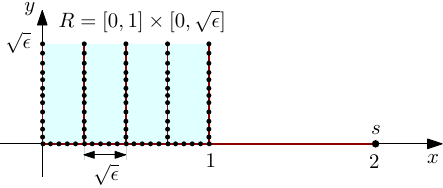}
    \caption{A source $s=(2,0)$ and points in a rectangle $R=[0,1]\times [0,\sqrt{\eps}]$.}
    \label{fig:LB2}
\end{figure}

It is clear from the construction that 
\[
    \wts(\mst(P\cup \{s\}) \leq 1+ \left\lfloor\frac{1}{\sqrt{\eps}}\right\rfloor\cdot \sqrt{\eps} \leq 2.
\]
Let $T$ be a minimum-weight Steiner $(1+\eps)$-ST 
for $P\cup \{s\}$. We show that $\wts(T) = \Omega(\eps^{-1/4})$.
For each point $p\in P$, the tree $T$ contains a $ps$-path of length at most $(1+\eps)\cdot d(p,s)$ in the ellipse $\mathcal{E}_{ps,\eps}$ of foci $p$ and $s$ and major axis $(1+\eps)\cdot d(p,s)$. 
By \Cref{lem:approxellipse}, we have $\mathcal{E}_{ps,\eps} \subset \mathcal{E}_{pb,2\eps}$, where $pb$ is a horizontal line segment and $d(p,b)\leq 2\cdot d(p,s)$. Consider the axis-aligned bounding box $B_p$ of the intersection of $\mathcal{E}_{pb,2\eps}$ with a vertical strip of width $\sqrt{\eps}$ whose left boundary contains $p$.  

We claim that the height of $B_p$ is $\Theta(\eps^{1/4})$. 
Indeed, for any $p\in P$, we have $1\leq d(p,s)\leq \sqrt{5}$, hence $1\leq d(b,p)\leq 2\sqrt{5}$. If we denote the width and height of $B_p$ by $\mathrm{width}(B_p)$ and $\mathrm{height}(B_p)$, respectively,
then \Cref{cor:spacing} implies 
$\mathrm{heigh}(B_p)
=\Theta(\sqrt{\mathrm{width}(B_p)}) 
=\Theta(\eps^{1/4})$, as claimed. 

Therefore, we can find $\Theta(\eps^{-1/4})$ points $p\in P$ in each vertical line such that the boxes $B_p$ are pairwise disjoint. 
Since $P$ is distributed on $\left\lfloor\frac{1}{\sqrt{\eps}}\right\rfloor$ vertical lines, 
we can find $\Theta(\eps^{-3/4})$ points $p\in P$ in each vertical line such that the boxes $B_p$ are pairwise disjoint.
By \Cref{lem:charging}, we have $\wts(T\cap B_p)\geq \mathrm{width}(B_p)=\sqrt{\eps})$. 
Summation over disjoint boxes $B_p$, yields 
$\wts(T)\geq \Theta(\eps^{-3/4})\cdot \sqrt{\eps} = \Theta(\eps^{-1/4})$, as required.
\end{proof}
\section*{Acknowledgments.}
Hung Le and Cuong Than are supported by the NSF CAREER award CCF-2237288, the NSF grants CCF-2517033 and CCF-2121952, and a Google Research Scholar Award. Cuong Than is also supported by a Google Ph.D. Fellowship.
Research by Csaba D.\ T\'oth was supported by the NSF award DMS-2154347.
Shay Solomon is funded by the European Union (ERC, DynOpt, 101043159). Views and opinions expressed are however those of the author(s) only and do not necessarily reflect those of the European Union or the European Research Council. Neither the European Union nor the granting authority can be held responsible for them. Shay Solomon is also funded by a grant from the United States-Israel Binational Science Foundation (BSF), Jerusalem, Israel, and the United States National Science Foundation (NSF). Work of Tianyi Zhang was done while at ETH Z\"urich when supported by funding from the starting grant ``A New Paradigm for Flow and Cut Algorithms'' (no. TMSGI2\_218022) of the Swiss National Science Foundation. 
\bibliographystyle{alphaurl}
\bibliography{ref}
\end{document}